\documentclass[11pt]{article}

\usepackage[T5]{fontenc}
\usepackage[margin=1in]{geometry}
\usepackage[labelfont=bf]{caption}
\usepackage{tabularray}

\usepackage{algorithm,algpseudocode,amsmath,amssymb,amsthm,float,mathtools,thmtools,thm-restate,scalerel,stackengine,verbatim,subcaption}
\usepackage[shortlabels]{enumitem}
\usepackage{tikz}
\usetikzlibrary{decorations.pathreplacing, positioning}
 
\stackMath
\newcommand\reallywidehat[1]{%
\savestack{\tmpbox}{\stretchto{%
  \scaleto{%
      \scalerel*[\widthof{\ensuremath{#1}}]{\kern-.6pt\bigwedge\kern-.6pt}%
          {\rule[-\textheight/2]{1ex}{\textheight}}
            }{\textheight}%
            }{0.5ex}}%
            \stackon[1pt]{#1}{\tmpbox}%
            }
\newcommand{\bool}{\set{0,1}}

\newcommand{\Cbb}{\mathbb{C}}
\newcommand{\Fbb}{\mathbb{F}}
\newcommand{\Ibb}{\mathbb{I}}
\newcommand{\Nbb}{\mathbb{N}}
\newcommand{\Rbb}{\mathbb{R}}
\newcommand{\Zbb}{\mathbb{Z}}

\newcommand{\mb}{\mathbf{m}}

\newcommand{\Gb}{\mathbf{G}}

\newcommand{\Xb}{\mathbf{X}}

\newcommand{\Ac}{\mathcal{A}}
\newcommand{\Bc}{\mathcal{B}}

\newcommand{\Fc}{\mathcal{F}}
\newcommand{\Gc}{\mathcal{G}}
\newcommand{\Lc}{\mathcal{L}}
\newcommand{\Mc}{\mathcal{M}}
\newcommand{\Sc}{\mathcal{S}}

\newcommand{\wt}[1]{\widetilde{#1}}

\newcommand{\wh}[1]{\widehat{#1}}
\newcommand{\eps}{\varepsilon}

\newcommand{\Uni}[1]{\mathcal{U}\left( #1 \right)}
\newcommand{\Bi}[2]{\text{\normalfont{Bi}}\paren*{#1, #2}}

\newcommand{\prob}[2][]{\operatorname*{\mathbb{P}}_{#1 }\brac*{#2}}
\newcommand{\E}[2][]{\operatorname*{\mathbb{E}}_{#1 }\brac*{#2}}

\newcommand{\bO}[1]{\operatorname*{O}\paren*{#1}}
\newcommand{\bOt}[1]{\operatorname*{\wt{O}}\paren*{#1}}
\newcommand{\lO}[1]{\operatorname*{o}\paren*{#1}}
\newcommand{\bOm}[1]{\operatorname*{\Omega}\paren*{#1}}
\newcommand{\lOm}[1]{\operatorname*{\omega}\paren*{#1}}
\newcommand{\bT}[1]{\operatorname*{\Theta}\paren*{#1}}

\newcommand{\uni}[1]{\operatorname*{\mathcal{U}\paren*{#1}}}

\DeclareMathOperator{\var}{\text{\normalfont{Var}}}
\DeclareMathOperator{\poly}{poly}

\DeclarePairedDelimiter\floor{\lfloor}{\rfloor}
\DeclarePairedDelimiter\ceil{\lceil}{\rceil}
\DeclarePairedDelimiter\norm{\lVert}{\rVert}
\DeclarePairedDelimiter\abs{|}{|}
\DeclarePairedDelimiter\brac{\lbrack}{\rbrack}
\DeclarePairedDelimiter\set{\lbrace}{\rbrace}
\DeclarePairedDelimiter\paren{\lparen}{\rparen}

\DeclareMathOperator{\tra}{Tr}

\newcommand{\NP}{\textsf{NP}}
\newcommand{\BQP}{\textsf{BQP}}
\newcommand{\QMA}{\textsf{QMA}}

\DeclarePairedDelimiter\bra{\langle}{|}
\DeclarePairedDelimiter\ket{|}{\rangle}
\DeclarePairedDelimiterX\braket[2]{\langle}{\rangle}{#1\delimsize\vert#2}

\newtheorem{theorem}{Theorem}
\newtheorem{lemma}[theorem]{Lemma}
\theoremstyle{definition}
\newtheorem{definition}[theorem]{Definition}

\theoremstyle{remark}
\newtheorem{remark}[theorem]{Remark}

\newcommand{\nth}{^\text{th}}

\newcommand{\tr}{^\text{tr}}

\newcommand{\dihp}{\text{\bf DIHP}}
\newcommand{\players}{T}
\newcommand{\yes}{\text{\bf YES}}
\newcommand{\no}{\text{\bf NO}}
\newcommand{\vlabels}{X}
\newcommand{\elabels}{w}
\newcommand{\mcut}{{\normalfont\textsc{Max-Cut}}}
\newcommand{\qmcut}{{\normalfont\textsc{Quantum}\ \mcut}}

\newcommand{\qubits}{\beta}

{\makeatletter
	\gdef\commentmark{%
		\expandafter\ifx\csname @mpargs\endcsname\relax 
		\expandafter\ifx\csname @captype\endcsname\relax 
		\marginpar{comment}
		\else
		comment 
		\fi
		\else
		comment 
		\fi}
	\gdef\comment{\@ifnextchar[\comment@lab\comment@nolab}
	\long\gdef\comment@lab[#1]#2{{\bf [\commentmark #2 ---{\sc #1}]}}
	\long\gdef\comment@nolab#1{{\bf [\commentmark #1]}}
}

\title{The Quantum and Classical Streaming Complexity of\\ Quantum and Classical Max-Cut}
\date{}
\author{John Kallaugher\\Sandia National Laboratories\\\texttt{jmkall@sandia.gov} \and Ojas Parekh\\Sandia National Laboratories\\\texttt{odparek@sandia.gov}}

\begin{document}
\maketitle

\begin{abstract}
\noindent
We investigate the space complexity of two graph streaming problems:
\textsc{Max-Cut} and its quantum analogue, \textsc{Quantum Max-Cut}.  Previous
work by Kapralov and Krachun \lbrack STOC~`19\rbrack\ resolved the
\emph{classical} complexity of the \emph{classical} problem, showing that any
$(2 - \varepsilon)$-approximation requires $\Omega(n)$ space (a
$2$-approximation is trivial with $\textrm{O}(\log n)$ space). We generalize
both of these qualifiers, demonstrating $\Omega(n)$ space lower bounds for $(2
- \varepsilon)$-approximating \textsc{Max-Cut}\ and \textsc{Quantum Max-Cut},
even if the algorithm is allowed to maintain a quantum state. As the trivial
approximation algorithm for \textsc{Quantum Max-Cut}\ only gives a
$4$-approximation, we show tightness with an algorithm that returns a $(2 +
\varepsilon)$-approximation to the \textsc{Quantum Max-Cut}\ value of a graph
in $\textrm{O}(\log n)$ space.  Our work resolves the quantum and classical
approximability of quantum and classical Max-Cut using $\textrm{o}(n)$ space.

We prove our lower bounds through the techniques of \emph{Boolean Fourier
analysis}. We give the first application of these methods to \emph{sequential
one-way quantum communication}, in which each player receives a quantum message
from the previous player, and can then perform arbitrary quantum operations on
it before sending it to the next. To this end, we show how Fourier-analytic
techniques may be used to understand the application of a quantum channel.
\end{abstract}

\thispagestyle{empty}
\setcounter{page}{0}

\newpage

\section{Introduction}

Quantum approaches for discrete optimization, such as the Quantum Approximate
Optimization Algorithm (QAOA) have received significant attention.  The seminal
work of Farhi, Goldstone, and Gutmann~\cite{farhi2014quantum} showed that QAOA
applied to an \NP-hard classical constraint satisfaction problem (CSP) gave a
better worst-case approximation than the best known classical approximation
algorithm at the time.  An improved classical approximation algorithm
subsequently followed~\cite{Barak2015Beating}; however, this seeded the
question of whether a quantum approximation algorithm might offer a provably
better approximation guarantee than the best classical approximation for some
CSP or discrete optimization problem, which still remains open.  One potential
barrier is that classical hardness of approximation results may also restrict
quantum approximation algorithms.  For example, it is generally not expected
that $\NP \subseteq \BQP$, so a quantum approximation is not expected to
overcome \NP-hardness of approximation results.  Even possibly weaker hardness
assumptions such as Unique-Games-hardness may impede quantum approximations.
It would be surprising if a quantum approximation were able to achieve a
$(1/0.878\ldots - \varepsilon)$-approximation\footnote{This result is more
typically stated as $0.878\ldots + \eps$, where an $\alpha$-approximation is
held to mean returning a value in $\brac{\alpha\cdot\text{OPT}, \text{OPT}}$,
for $\text{OPT}$ the correct value. However, we follow previous work on
streaming \mcut\ by instead using a $K$-approximation to mean returning a value
in $\brac{\text{OPT}, K\cdot\text{Opt}}$.} for the Maximum Cut Problem (\mcut),
which is Unique-Games-hard~\cite{KhotKMO07}.  

Although the prospects for quantum approximations for classical CSPs may seem
limited, a natural question is whether quantum approximations can offer
provably better guarantees for \emph{quantum versions} of CSPs.  The $k$-Local
Hamiltonian Problem ($k$-LH) serves as the canonical \QMA-hard quantum
generalization of $k$-CSP.  A recent line of work has enjoyed success in
devising nontrivial classical approximations for $2$-LH~\cite{G12, B16, B19, GharibianParekh19, PT20, Anshu2020, PT21, A21}; however, truly
quantum approximations for LH remain elusive.  Hardness of approximation
results with respect to \QMA\ are even more elusive, as the existence of a
quantum analogue of the classical PCP theorem, a cornerstone for hardness of
approximation, remains a major open
question~\cite{aharonov2013guest,natarajan2018low}.

We seek to understand the power of quantum versus classical approximations for
$2$-CSP and $2$-LH in the streaming setting, where \emph{space} is the
computational quantity of interest.  In particular we consider the \mcut\ (MC)
and \qmcut\ (QMC) problems.  \mcut\ is a prototypical CSP in the sense that
approximation and hardness results are typically devised for \mcut\ and then
generalized to other CSPs; \qmcut\ has emerged to serve a similar role in
approximating $2$-LH. \qmcut\ is also closely related to the quantum Heisenberg model (see~\cite{GharibianParekh19}), which is a well-studied model of quantum magnetism introduced in the late 1920s.

For \emph{classical} algorithms applied to \emph{classical} $\mcut$, tight bounds\footnote{Up to log factors in the space
complexity, as is typical for streaming algorithms.} for the space complexity
in terms of the approximation factor are known~\cite{KK19}---our work
generalizes these results in both ways, giving tight bounds on the
approximation factor attainable in $\lO{n}$ space by quantum streaming
algorithms for classical $\mcut$ and by quantum and classical algorithms for
$\qmcut$.  

We find, perhaps surprisingly, that quantum streaming algorithms offer no
advantage over classical ones in approximating \mcut\ or \qmcut.  Although our
main contribution is a quantum hardness result, the matching upper bound for
approximating \qmcut\ in the stream requires analyzing a nontrivial streaming
algorithm, which is a departure from the case of \mcut.

\subsection{Our Contributions}
Ours is the first work to consider streaming versions of $2$-LH or any kind of
quantum optimization problem.  Just as the results of~\cite{KK19} have been
expanded for more general CSPs, we expect that our results for \qmcut\ will
apply to more general instances of $2$-LH.  Indeed there is precedent for this
in the standard approximation setting~\cite{H20,PT20}.    

We give tight (up to an arbitrarily small additive constant in the
approximation factor) characterizations of the best possible approximation
factor achievable in $\lO{n}$ space for quantum and classical algorithms for
quantum and classical Max-Cut. Our results are laid out in
Table~\ref{tbl:results}.

\begin{table}
\centering
\begin{tblr}{
  hline{1} = {2-7}{solid},
  hline{2,3,5} = {solid},
  vline{1} = {2-5}{solid},
  vline{2,4,6} = {solid}, 
  cells = {c},
  cell{1}{2,4} = {c=2}{c}, 
  hspan = even, 
}
          & $\mcut$ &    & \qmcut &    \\
  Approximation Factor & $2 + \eps$ & $2-\eps$  & $2+\eps$ & $2 - \eps$  \\
	Classical Algorithm & $\bO{\log n}$ & $\bOm{n}$  & {\boldmath $\bO{\log n}$} & {\boldmath
$\bOm{n}$} \\
  Quantum Algorithm & $\bO{\log n}$ & {\boldmath $\Omega(n)$} & {\boldmath $\bO{\log n}$}                    & {\boldmath $\bOm{n}$} \\
\end{tblr}
\caption{The space needed by quantum and classical algorithms for quantum and
classical Max-Cut. Results from this paper are shown in bold.}
\label{tbl:results}
\end{table}

Our lower bounds are encompassed in the following theorem (the classical lower
bound for $\qmcut$ is a special case of this, as any classical streaming
algorithm can be implemented as a quantum streaming algorithm).
\begin{restatable}{theorem}{maxcutlb}
\label{thm:maxcutlb}
For any $\varepsilon > 0$, any quantum streaming algorithm for \mcut\ or
\qmcut\ that returns a $(2-\varepsilon)$-approximation with probability $2/3$
requires \[
n/2^{\bO{1/\varepsilon^2}}
\]
qubits of storage.
\end{restatable}

For $\mcut$ the upper bound for $(2 + \eps)$-approximation (in fact, even
$2$-approximation) is trivial, as a graph on $m$ edges always has $\mcut$ value
between $m/2$ and $m$. However, for $\qmcut$ the trivial approximation is only
a $4$-approximation, so we give an algorithm that returns a $(2 +
\eps)$-approximation using $\bO{\log n}$ space. We also give an algorithm for
\emph{weighted} graphs, but in this case we are only able to attain a $(5/2 +
\eps)$-approximation (the lower bound remains a 2-approximation).
\begin{restatable}{theorem}{qmcub}
\label{thm:qmcub}
Let $G$ be a weighted graph on $n$ vertices with weights that are multiples of
$1/\poly(n)$. Then for any $\varepsilon, \delta \in (0,1)$ there is a streaming
algorithm that returns a $(5/2 + \varepsilon)$-approximation to the \qmcut\ value of $G$ with probability at least $1 - \delta$ using
$\bO{\frac{1}{\eps^2}\log \frac{1}{\delta}\log n}$ space. If all the weights in
the graph are $1$, it returns a $(2 + \varepsilon)$-approximation instead.
\end{restatable}
We note here two lacunae in our results for (unweighted) graphs. Firstly, for
classical \mcut\ it is possible to achieve a $(1 + \eps)$-approximation in
$\bOt{n}$ space, through the use of cut-preserving sparsifiers~\cite{anh-guha}.
However, analogous results on sparsifiers for general $2$-local Hamiltonians
are not known, and indeed there are results pointing in the opposite
direction~\cite{AZ19}.  So while we can characterize the approximation factors
possible in sublinear space the semi-streaming complexity remains open.
Secondly, our $\bO{\log n}$-space upper bound for $\qmcut$ only gives a $(2 +
\eps)$-approximation instead of a $2$-approximation. This is a consequence of
the fact that it is based on graph parameters that must themselves be
approximated rather than just the number of edges, which can be calculated
exactly.

\paragraph{Fourier Analysis for Quantum Channels} The technical core of our
lower bound is a quantum communication complexity bound for a sequential
one-way communication problem (originally introduced in~\cite{KK19} in the
classical setting), in which the first player sends a message to the second
player, the second to the third, and so on. Our bound for this problem is based
on a novel application of \emph{Boolean Fourier analysis}\footnote{For a
general overview of Boolean Fourier analysis, see~\cite{O14}.}---in particular,
we prove that a key inequality associated with this problem, analyzed
in~\cite{KK19} for the classical case, is preserved even in the presence of
quantum communication.

The application of Boolean Fourier Analysis to two-player one-way communication
problems in the classical setting goes back to~\cite{GKKRd07}, in which it was
used to prove lower bounds for the Boolean Hidden Matching problem (and its
application to communication complexity more generally goes back further,
e.g.~\cite{R95,K07}). This problem, and its generalization in the Boolean
Hidden Hypermatching problem (analyzed in~\cite{VY11}), are the main route by
which Fourier analysis has contributed to lower bounds for streaming
algorithms.

However, in later years these techniques have been extended to communication
lower bounds (and corresponding streaming lower bounds) with different
configurations of players. In~\cite{KKP18} they were applied to problems where
many players communicate with a single referee, while in~\cite{KKSV17} they
were extended to problems where players communicate in a line, as is the case
in the Distributed Implicit Partition Problem (from~\cite{KK19}) we make use of
in this paper. In~\cite{CGSVV22}, a generalization of this problem was studied
through the use of Fourier analysis on $\Zbb_q^n$.

The core ingredient of most of these lower bounds is a hypercontractive
Fourier coefficients lemma from~\cite{KKL88}, that can be seen as generalizing
facts about sampling protocols, in which a player chooses some subset of their
input to send to protocols where players send arbitrary (classical) messages.
This lemma was generalized to \emph{matrix}-valued functions in~\cite{BARdW08},
opening the door to the application of Fourier-analytic methods to lower bounds
for \emph{quantum} communication protocols, as these can be seen as functions
from inputs to density matrices.

This was first used to prove quantum lower bounds on the complexity of the
Boolean Hidden Hypermatching problem~\cite{SW12}. This result was further
generalized in~\cite{DM20}, while~\cite{AD21} generalized the Fourier
coefficients lemma further, in order to obtain quantum lower bounds for $\mcut$
and more general hypergraph problems. However, these are all two-player one-way
communication problems. We give the first application of Fourier-analytic
techniques to \emph{sequential quantum one-way communication}. The key
technical challenge is in finding methods for applying Fourier analysis to the
application of \emph{quantum channels}.

In classical communication, as long as we consider a single ``hard'' input
distribution, a player's message can be without loss of generality assumed to
be a deterministic function of the message they received and their input. In
sequential quantum communication, however, the player may apply an arbitrary
quantum channel to the message they receive. Our key insight is that, as
quantum channels are linear operators, many of the techniques of Fourier
analysis, including the convolution lemma for Fourier coefficients, may be
applied to them.

\subsection{Other Related Work}

\paragraph{Streaming bounds for \mcut} The fact that that a $2$-approximation
for \mcut\ is possible in $\bO{\log n}$ space is an immediate consequence of
the fact that the $\mcut$ value is always at least $m/2$. Less immediately, but
still a consequence of standard results in streaming algorithms, is the fact
that it can be $(1 + \eps)$-approximated in $\bOt{n}$ space, through
sparsifiers that preserve cut values~\cite{anh-guha}.  The question, then, was
whether a better approximation than the first could be attained in less space
than the second.

In~\cite{KK15,KKS15}, it was shown that any $(2 - \varepsilon)$-approximation
would require at least polynomial space in $n$, while~\cite{KKSV17} showed that
$(1 + \eps)$-approximation would require $\bOm{n}$ space. This left open the
possibility of intermediate results, but~\cite{KK19} closed the door on this
possibility, proving that $(2 - \eps)$-approximation would require $\bOm{n}$
space for any constant $\eps > 0$.

However, the above results are only for classical algorithms. In~\cite{AD21}, a
polynomial lower bound was shown that applies even to quantum streaming
algorithms, but this left open the possibility that a $(2-\eps)$ approximation
was possible in $\lO{n}$ space for quantum algorithms.

\paragraph{Quantum streaming algorithms}
The first work on quantum streaming was~\cite{LG06}, which showed that there
are problems that that are exponentially easier for quantum streaming
algorithms than classical ones. In~\cite{GKKRW07}, it was shown that this is
true even for a function that does not depend on the order of the stream (the
more ``standard''  streaming model).

Later work has investigated the question of whether quantum streaming can
obtain advantages over classical for problems of independent classical interest
(as the aforementioned work is for problems constructed for the purpose of
proving separations). The problem of recognizing {\tt Dyck}(2) in the stream
was considered as a candidate problem in~\cite{JN14,NT17}, but only negative
results were found. For problems where $\lOm{1}$ passes are allowed over the
stream, \cite{M16} and~\cite{HM19} showed an advantage for the well-studied
moment estimation problem. Later,~\cite{K21} showed that an advantage exists
in the one-pass setting for the problem of counting triangles in graph streams.

\paragraph{Approximating \qmcut}
\qmcut\ was introduced in~\cite{GharibianParekh19}, where a classical
$1/0.498$-approximation algorithm was given, akin to the Goemans-Williamson
algorithm for \mcut, that produces an unentangled product state.  Since the gap
between the best product state and best entangled quantum state on a single
edge is two (see Section~\ref{sec:qmcut-and-mcut}), at best a $2$-approximation
is possible for algorithms that return product states, and so the
$1/0.498$-approximation is nearly optimal among such algorithms.  By rounding
to entangled states, \cite{Anshu2020} gave the first approximation with
guarantee better than 2.  Subsequently~\cite{PT21} showed how to use higher
levels of the quantum Lasserre hierarchy of semidefinite programs to obtain a
slight improvement over~\cite{Anshu2020}.

\section{Proof Overview}
\subsection{Lower Bounds}
Our lower bounds for the quantum streaming complexity of $\mcut$ and $\qmcut$
are derived from a new analysis of the Distributional Implicit Hidden Partition
(\dihp) problem introduced in~\cite{KK19} to prove lower bounds for the
streaming complexity of approximating classical $\mcut$. We restate this
problem here.

\subsubsection{The Distributed Implicit Hidden Partition Problem}
In an instance of $\dihp(n, \alpha, T)$, $T$ players are each given a partial
matching $M_t$ of $\alpha n$ edges on $n$ vertices, with each edge labelled
with a bit. Either these bit labels are generated by choosing a random
partition of $\brac{n}$ and assigning 1 to the edges crossing the partition (a
$\yes$ case) or they are chosen uniformly at random (a $\no$ case).

The players are allowed one-way communication, from player $i$ to player $i+1$
for each $i$, and are additionally given the matching edges (but not the edge
labels) of every previous player for free. Their goal is to determine whether
their inputs were drawn from a $\yes$ case or a $\no$ case with probability at
least $2/3$ over the random draw and any internal randomness they may use.

\paragraph{Reduction to Classical Max-Cut} If each player $t$ creates the graph
$G_t$ consisting of edges labelled $1$ in $M_t$, $G = \bigcup_t G_t$ will be
bipartite in a $\yes$ case, and close to random in a $\no$ case. This means it
is possible to cut every edge in the first case, and not much more than half of
them in the second.  Therefore, an algorithm that returns a $(2 -
\varepsilon)$-approximation to $\mcut$ can distinguish them if $\varepsilon$ is
large enough (by making $\alpha$ small enough and $T$ large enough, we can make
the necessary $\varepsilon$ arbitrarily small). Therefore, a $\mcut$ algorithm
using $S$ space gives a protocol in which each player sends a size-$S$ message,
by having each player run the algorithm on their input and then send their
algorithm's state to the next player. The graphs the players get with this
reduction are illustrated in Figure~\ref{fig:dihp_reduction}.
\begin{figure}
\centering
\begin{subfigure}[t]{\textwidth}
\centering
\tikzset{basevertex/.style={shape=circle, line width=0.5, minimum size=4pt,
inner sep=0pt, draw}} \tikzset{defaultvertex/.style={basevertex, fill=blue!70}}
\begin{tikzpicture}[
arrow/.style={single arrow, thick,draw=blue!70,fill=blue!30, minimum
height=10mm}, scale = 1]
\foreach \i in {0,1,2} {
\pgfmathsetmacro{\offset}{5 * \i}
\node[style=defaultvertex,fill=red] (v0\i) at (\offset ,0) {};
\node[style=defaultvertex,fill=red] (v1\i) at (\offset +1 ,0.2) {};
\node[style=defaultvertex] (v2\i) at (\offset +0.5 ,0.8) {};
\node[style=defaultvertex] (v3\i) at (\offset  ,0.9) {};
\node[style=defaultvertex] (v4\i) at (\offset + 1.1  ,1) {};
\node[style=defaultvertex,fill=red] (v5\i) at (\offset + 0.6, -0.31) {};
}
\draw (v00) -- (v20);
\draw (v10) -- (v30);
\draw (v01) -- (v31);
\draw (v51) -- (v41);
\draw (v12) -- (v32);
\draw (v52) -- (v42);
\end{tikzpicture}
\caption{In a $\yes$ case, only edges crossing the underlying partition are included.}
\end{subfigure}
\begin{subfigure}[t]{\textwidth}
\centering
\tikzset{basevertex/.style={shape=circle, line width=0.5, minimum size=4pt,
inner sep=0pt, draw}} \tikzset{defaultvertex/.style={basevertex, fill=blue!70}}
\begin{tikzpicture}[
arrow/.style={single arrow, thick,draw=blue!70,fill=blue!30, minimum
height=10mm}, scale = 1]
\foreach \i in {0,1,2} {
\pgfmathsetmacro{\offset}{5 * \i}
\node[style=defaultvertex] (v0\i) at (\offset ,0) {};
\node[style=defaultvertex] (v1\i) at (\offset +1 ,0.2) {};
\node[style=defaultvertex] (v2\i) at (\offset +0.5 ,0.8) {};
\node[style=defaultvertex] (v3\i) at (\offset  ,0.9) {};
\node[style=defaultvertex] (v4\i) at (\offset + 1.1  ,1) {};
\node[style=defaultvertex] (v5\i) at (\offset + 0.6, -0.31) {};
\draw (v50) -- (v40);
\draw (v00) -- (v20);
\draw (v01) -- (v31);
\draw (v51) -- (v41);
\draw (v22) -- (v32);
\draw (v02) -- (v12);
}
\end{tikzpicture}
\caption{In a $\no$ case, each edge received is included with probability $\frac{1}{2}$.}
\end{subfigure}
\caption{The graphs each player receives when reducing $\dihp$ to $\mcut$.}
\label{fig:dihp_reduction}
\end{figure}
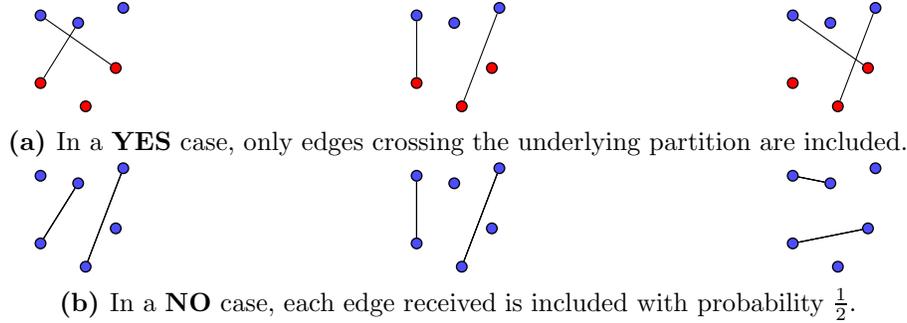

One problem with this is that, if the players' matchings are randomly chosen
they may share edges. Our approach to this differs somewhat from that
of~\cite{KK19}. Instead of considering multigraphs, we take advantage of the
fact that player $t$ is allowed to know the matching (but not the edge labels)
of players $s < t$. This means we can have them decline to add edges that are
present in previous matchings, guaranteeing that the final graph is simple. We
show that, as the number of edges thus removed is small, it has little effect
on the reduction.

\paragraph{Extending the Reduction to Quantum Max-Cut} The reduction to Quantum
Max-Cut uses exactly the same mapping from $\dihp$ instances to graphs. We
consider the following SDP, \[
\max_{f : V \rightarrow S^{n-1}} \sum_{uv \in E} - \langle f(u), f(v) \rangle
\]
which is a shifted version of the standard Goemans-Williamson SDP for $\mcut$.
In particular, its optimal value is an upper bound on $2K - m$, where $K$ is
the $\mcut$ value of a graph. Usefully, when $2K - m$ is small, a converse
property holds, as the optimal value of this SDP is at most a constant factor
times larger than $2K - m$~\cite{CW04}. This means the graphs generated by
$\no$ instances of $\dihp$ will have small values of this SDP.

This gives us a $\qmcut$ lower bound, because this SDP \emph{also} upper bounds
$\frac{4}{3}Q - \frac{m}{3}$, where $Q$ is the $\qmcut$ value of the
graph\footnote{See Section~2.3 of \cite{hwang2021unique}. Note that in the
cited work, both $\qmcut$ and the SDP are scaled by $\frac{1}{m}$ relative to
our usage.}.  So $\no$ instances will create graphs with $\qmcut$ value
approximately $m/4$.  Conversely $\yes$ instances will create graphs with
$\qmcut$ value at least $m/2$, as they are bipartite and the $\qmcut$ value is
always at least half the $\mcut$ value. So a $(2-\eps)$ approximation algorithm
would suffice to distinguish between the two.

\subsubsection{Quantum Communication Lower Bounds for DIHP} In~\cite{KK19} it
was shown that $\dihp$ is hard when the players are only allowed to send
classical messages, requiring $\bOm{n}$ space when $\alpha$ and $T$ are
constant. The majority of the technical difficulty of our lower bounds is in
proving that $\dihp$ is hard even if the players are allowed to send
\emph{quantum} messages. This immediately implies that quantum algorithms must
use $\bOm{n}$ space to $(2-\eps)$-approximate $\mcut$ or $\qmcut$, and so no
quantum advantage for either problem is possible.

\paragraph{Reduction to Boolean Fourier Analysis} As with the classical lower
bound of~\cite{KK19}, our proof depends on applying Fourier analysis to
functions on the Boolean cube. In particular, we will show that a bound on
Fourier coefficients used in the classical proof is maintained even in the
presence of quantum communication. We start by providing an intuition for the
significance of this bound.

Suppose the game is in a $\yes$ case, and so player $t$'s input depends only on the
matching $M_t$ and the randomly chosen partition (which we may write $x \in
\bool^n$, with the bit of vertex $i$ determining which side of the partition it
is on). Then, fixing $(M_s)_{s=1}^t$, we can write a function \[ 
f_t : \bool^n \rightarrow \Cbb^{2^\qubits \times 2^\qubits} 
\] where $f(x)$ is the density matrix sent by
player $t$ if the partition is $x$, and $\qubits$ is the number of qubits used to
represent that state.

Now suppose player $t+1$ would like to determine whether they are in a $\yes$
or a $\no$ case. They have received $f_{t}(x)$ if they are in a $\yes$ case,
and they want to determine if it is consistent with being in a $\yes$ case. In
addition, they have the bit labels of the edges in $M_{t+1}$. Therefore, for
any odd-cardinality set of edges in $M_{t+1}$, they know the parity of the set
of vertices in $x$ matched by these edges.  We write such sets of vertices as
$M_{t+1}\tr s$ for a string $s \in \bool^{\alpha n}$ indexing a subset of the
edges in $M_{t+1}$.

Now suppose the player looked at only one of these sets $s$, and so knew the
parity of the vertices $M_{t+1}\tr s$ alone. To tell whether $f_t(x)$ could
come from a $\yes$ instance, they need\footnote{We are eliding the possibility
that, for instance, the state player $t+1$ receives is impossible or unlikely
in a $\yes$ case due to, for instance, only arising if a triangle in previously
arrived edges has every edge labelled 1. However it turns out this possibility
is already accounted for by considering what a previous player would've seen on
receiving the third edge of that triangle.} its average value when the parity
of $M_{t+1}\tr s$ is $0$ to be distinguishable from its average value when the
parity of $M_{t+1}\tr s$ is $1$.

The distinguishability of two distributions over quantum states is given by the
trace norm of the difference between their density matrices, so the quantity
the player would need to be large is
\begin{align*}
\frac{1}{2}\norm*{\frac{1}{2^{n-1}}\sum_{\substack{x \in \bool^n:\\ x \cdot M_{t+1}\tr s = 0}} f_t(x) -
\frac{1}{2^{n-1}}\sum_{\substack{x \in \bool^n:\\ x \cdot M_{t+1}\tr s = 1}}f_t(x)}_1 &=
\frac{1}{2^n}\norm*{\sum_{x \in \bool^n} f_t(x) (-1)^{x \cdot M_{t+1}\tr s}}_1\\
&= \norm{\wh{f}_t (M_{t+1}\tr s)}_1
\end{align*}
where we now introduce $\wh{f}_t$, the \emph{Fourier transform} of $f_t$, given by \[
\wh{f}(S) = \frac{1}{2^n}\sum_{x \in \bool^n} f(x)(-1)^{S \cdot x}\text{.}
\]
It turns out that this sums nicely---it can be shown that player $T$'s ability
to distinguish between a $\yes$ and a $\no$ case is bounded by \[
\sum_{t=1}^T\sum_{s \in \bool^{\alpha n}\setminus\set{\emptyset}}
\norm{\wh{f}_t(M_{t+1}\tr s)}_1
\]
and so our goal will be to prove that this sum is small in expectation over
$(M_t)_{t=1}^T$.

To prove this, we bound the total value of weight-$2\ell$ Fourier coefficients
for every $\ell$. As a ${\sim}\binom{\alpha n}{\ell}/\binom{n}{2\ell}$ fraction
of these will end up being matched by a set of $\ell$ matching edges, it
suffices to prove that the value is bounded by\footnote{This expression changes
somewhat when $\ell \ge \qubits$, but we will disregard those highest-order
terms in this overview.} \[
\paren*{\frac{\sqrt{\qubits n}}{\ell}}^\ell
\]
where we have dropped some constants exponential in $T$ and $\ell$. Then if
$\qubits \ll n$, this expression will be small enough for the final states to
be hard to distinguish.

\paragraph{The Evolution of Fourier Coefficients} We will bound the expression
above by induction on $t$, considering how these coefficients evolve based on
the message sent from player $t$ to player $t+1$. This is where the
\emph{quantum} difficulty of the proof will arise, and is the most important
novel element in our analysis---the combinatorial aspects of the evolution are
similar to those in the classical case but now player $t$ may apply a quantum
channel to generate $f_t$ rather than sending a
deterministic\footnote{When proving a lower bound for a classical communication
problem with a known input distribution, one may without loss of generality
assume the players act deterministically.} message based on their input and
the message $f_{t-1}$.

The base case of the induction is straightforward (for simplicity we can think
of player 1 as receiving $0^\qubits$ from a player $0$, and consider only an
inductive step).  For the inductive step, we need to understand the effect of
player $t$ applying a quantum channel $\Ac$ to $f_{t -1}(x)$. This quantum
channel itself is determined by player $t$'s input, and therefore (again fixing
$(M_s)_{s=1}^t$) we can write $\Ac_x$ for its value when the underlying
partition is $x$. As quantum channels are linear operators, we can define a
Fourier transform \[
\wh{\Ac}_S = \frac{1}{2^n}\sum_{x \in \bool^n}\Ac_x (-1)^{S \cdot x}
\]
that in particular obeys the convolution lemma for the Boolean Fourier
transform, which tells us that \[
\wh{\Ac_x f_{t-1}(x)}(S) = \sum_U \wh{\Ac}_U \wh{f}_{t-1}(U \oplus S)\text{.}
\]
Using the fact that $\wh{\Ac}_U$ is $0$ whenever $U$ is not $M_ts$ for some $s
\in \bool^{\alpha n}$ (intuitively, this is because $\Ac_x$ only depends on the
edge labels of the edges in $M_t$), we can write down ``mass transfer'' lemmas
describing how coefficients of weight $2\ell_2$ of $f_t$ are formed from
coefficients of weight $2\ell_1$ of $f_{t-1}$. We want to know how much weight
can be contributed to $\wh{\Ac_x f_{t-1}(x)}(S)$ from $\wh{\Ac}_U
\wh{f}_{t-1}(M_t\tr s \oplus S)$ where $\abs{S} = \ell_2$ and $M_t\tr s \oplus
S = \ell_1$.

We can think of this in terms of three more parameters, $a$ the number of edges
from $M_t\tr s$ that are entirely contained in $S$, $b$ the number of edges that
each have one endpoint in $S$, $c$ the number that are entirely outside of $S$
(so $\ell_2 = \ell_1 - a + c$). We end up with the amount of ``mass''
transferred from $\ell_1$-weight coefficients via $M\tr_t s$ with this property
being bounded by \[
\sum_{\substack{S \in \bool^{n}\\ \abs{S} =
2\ell_1}} \sum_{\substack{u \in \bool^{\alpha n}\\\abs{u} =
a}}\sum_{\substack{v \in \bool^{\alpha n}\\ \abs{v} = b}}
I_S(u)B_S(v)\sum_{\substack{w \in \bool^{\alpha n}\\ \abs{w} = c}}
\norm{\wh{\Ac}^t_{M_t\tr(u \oplus v \oplus w)}\wh{f}_{t-1}(S)}_1
\]
where $I_S(u)$ and $B_S(v)$ are indicator variables on whether $M_t\tr u$ is
entirely contained in $S$ and $M_t\tr v$ has one endpoint of each edge in $S$.
See Figure~\ref{fig:evolution} for an illustration.
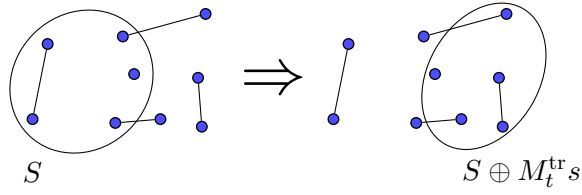
\begin{figure}
\centering
\tikzset{basevertex/.style={shape=circle, line width=0.5, minimum size=4pt,
inner sep=0pt, draw}} \tikzset{defaultvertex/.style={basevertex, fill=blue!70}}
\begin{tikzpicture}[
arrow/.style={single arrow, thick,draw=blue!70,fill=blue!30, minimum
height=10mm}, scale = 1]
\node[style=defaultvertex] (v1) at (0,0) {};
\node[style=defaultvertex] (v2) at (-0.2,-1) {};
\node[style=defaultvertex] (v3) at (1,0.1) {};
\node[style=defaultvertex] (v4) at (1.15,-0.4) {};
\node[style=defaultvertex] (v5) at (0.9,-1.05) {};
\draw (0.45, -0.5) circle [x radius=0.9, y radius=1, rotate=315];
\node at (-0.2, -1.7) {$S$};
\node[style=defaultvertex] (v6) at (2.1, 0.4) {};
\node[style=defaultvertex] (v7) at (2, -0.45) {};
\node[style=defaultvertex] (v8) at (2.05,-1.1) {};
\node[style=defaultvertex] (v9) at (1.5,-1) {};
\draw (v1) -- (v2);
\draw (v3) -- (v6);
\draw (v7) -- (v8);
\draw (v5) -- (v9);

\node at (3,-0.5) {\Huge $\Rightarrow$};
{
\newcommand{\offset}{4}
\node[style=defaultvertex] (v1) at (\offset + 0,0) {};
\node[style=defaultvertex] (v2) at (\offset + -0.2,-1) {};
\node[style=defaultvertex] (v3) at (\offset + 1,0.1) {};
\node[style=defaultvertex] (v4) at (\offset + 1.15,-0.4) {};
\node[style=defaultvertex] (v5) at (\offset + 0.9,-1.05) {};
\node[style=defaultvertex] (v6) at (\offset + 2.1, 0.4) {};
\node[style=defaultvertex] (v7) at (\offset + 2, -0.45) {};
\node[style=defaultvertex] (v8) at (\offset + 2.05,-1.1) {};
\node[style=defaultvertex] (v9) at (\offset + 1.5,-1) {};
\node at (\offset + 2.3, -1.7) {$S \oplus M_t\tr s$};
\draw (\offset + 1.8, -0.43) circle [x radius=0.74, y radius=1.04, rotate=330];
\draw (v1) -- (v2);
\draw (v3) -- (v6);
\draw (v7) -- (v8);
\draw (v5) -- (v9);
}
\end{tikzpicture}
\caption{When player $t$ receives the matching $M_t$, each subset $s$ of the
edges in $M_t$ and each Fourier coefficient $\wh{f}_{t-1}(S)$ corresponds to a
new Fourier coefficient $\wh{f}_t(S \oplus M_t\tr s)$. In this example $s$
includes $a = 1$ edge internal to $S$, $b = 2$ edges with one endpoint in $S$,
and $c = 1$ edge outside, so the resulting coefficient $S \oplus M_t\tr s$ has
weight $\abs{S} + a - c = 5$.}
\label{fig:evolution}
\end{figure}

The final tool we need to bound this is an extension of the matrix-valued
Fourier coefficients inequality, a consequence of Theorem~1 of~\cite{BARdW08}
(itself a generalization of a lemma of~\cite{KKL88}) that has previously been
used for two-player quantum lower bounds~\cite{SW12,DM20,AD21}.  This will tell
us that\footnote{Again, this has to be somewhat changed when $c$ is
particularly large.} \[
\sum_{\substack{w \in \bool^{\alpha n}\\ \abs{w} = c}}
\norm{\wh{\Ac}^t_{M_t\tr(u \oplus v \oplus w)}\wh{f}_{t-1}(S)}_1 =
\binom{\bO{\qubits}}{c}\norm*{\wh{f}_{t-1}(S)}_1\text{.} 
\] 
With this in place, and using the fact that \[
\prob{I_S(u)B_S(v)} \sim \frac{\binom{\ell_1}{a}}{\binom{n}{a}} \cdot \frac{\binom{\ell_1}{b}}{\binom{n}{b}}
\]
we can bound the above in expectation over $M_t$, and from then the proof
becomes an exercise in carefully evaluating sums.

\subsection{Space Upper Bounds for Quantum Max-Cut}
For classical $\mcut$ a trivial classical algorithm achieving a
$2$-approximation in logarithmic space is already known---count the number of
edges (or total weight for a weighted graph) $m$ and report $m$, which is at
most twice the true value. As our lower bound for quantum algorithms for
classical $\mcut$ is the same as the classical one, nothing more is needed
here. However, for $\qmcut$ the story is a bit different. The trivial lower
bound in this case is $m/4$, and so the aforementioned algorithm would only
guarantee a 4-approximation.

We give a simple algorithm that achieves a $(2 + \eps)$-approximation in the
unweighted case, and a $(5/2 + \eps)$-approximation in the weighted case. The
basic idea will be the same in both cases, so for ease of exposition the rest
of the discussion in this section will assume a weighted graph, and we will
point out where every edge having unit weight allows a better approximation.

\paragraph{Upper Bounding the QMC Value} Let $m$ be the total weight of the
graph, and let $W = \sum_{u \in V} \max_{v \in N(u)}w_{uv}$, the sum of the
max-weight edges incident to each vertex (so $W$ is just the number of
non-isolated vertices in the unweighted case). It is known~\cite{Anshu2020}
(see Lemma~\ref{lem:qmcut-bound})
that \[
\frac{m}{2} + \frac{W}{4}
\]
is an upper bound for $\qmcut$. So we want lower bounds in
terms of $m$ and $W$.

\paragraph{Lower Bounding the QMC Value in General Weighted Graphs} We use a
modified version of an argument of~\cite{Anshu2020}.  Consider the subgraph
formed by taking the highest-weight edge incident to every vertex. We can
decompose this into a matching $M$ consisting of every edge ``chosen'' by two
vertices, and a forest $F$ of all the other edges (note that the two together
are also a forest). Abusing notation to use the names of the objects to also
denote their total weights, we have $2M + F = W$. 

Now, for any edge, it is possible to earn $\qmcut$ energy equal to its weight
by assigning its vertices the singlet. Secondly, when we have a collection of
vertex-disjoint graphs it is possible to maximize each of their $\qmcut$
energies separately and still earn energy $w_e/4$ for each edge $e$
\emph{between} distinct pairs of graphs. So there is a solution earning $M +
(m - M)/4$.

Secondly, as $M \cup F$ is a forest, cutting it clasically earns energy
$(M + F)/2$, as any classical cut gives a quantum cut earning at least half as
much energy. By minimizing these two expressions subject to $2M + F = W$ it can
be shown that the $\qmcut$ value is at least \[
\frac{m}{5} + \frac{W}{10}
\]
giving a $(5/2)$-approximation determined only by $m$ and $W$. We illustrate
this construction in Figure~\ref{fig:tree_and_matching}.
\begin{figure}
\centering
\tikzset{basevertex/.style={shape=circle, line width=0.5, minimum size=4pt,
inner sep=0pt, draw}} \tikzset{defaultvertex/.style={basevertex, fill=blue!70}}
\begin{tikzpicture}[
arrow/.style={single arrow, thick,draw=blue!70,fill=blue!30, minimum
height=10mm}, scale = 1]
\node[style=defaultvertex] (v0) at (0,0) {};
\node[style=defaultvertex] (v1) at (1,0.5) {};
\node[style=defaultvertex] (v2) at (1,-0.5) {};
\node[style=defaultvertex] (v3) at (2,0.5) {};
\node[style=defaultvertex] (v4) at (2,-0.5) {};
\node[style=defaultvertex] (v5) at (3,0.5) {};
\node[style=defaultvertex] (v6) at (2,-1) {};
\draw (v0) -- (v1);
\draw (v0) -- (v2);
\draw (v1) -- (v3);
\draw (v2) -- (v6);
\draw[color=red,line width=1pt] (v3) -- (v5);
\draw[color=red,line width=1pt] (v2) -- (v4);
\draw[dashed] (v1) -- (v2);
\draw[dashed] (v3) -- (v4);
\node at (0.4, 0.45) {$4$};
\node at (0.4, -0.45) {$3$};
\node at (1.5, 0.7) {$5$};
\node at (2.5, 0.7) {$6$};
\node at (1.5, -0.3) {$4$};
\node at (1.4, -0.95) {$2$};
\node at (1.15, 0) {$1$};
\node at (2.15, 0) {$1$};
\end{tikzpicture}
\caption{Proving a lower bound for the optimal \qmcut{} value of a weighted
graph. The edges ``chosen'' by one vertex (in solid black) form a tree $T$,
while the edges ``chosen'' by two (in solid red) form a matching $M$. There is an
assignment earning energy $\frac{T + M}{2} = 12$ from assigning a perfect
classical cut to $T \cup M$, and one earning energy $M + \frac{M - m}{4} = 14$
from assigning the singlet to every edge in $M$ and earning $w_e/4$ on every
other edge.}
\label{fig:tree_and_matching}
\end{figure}
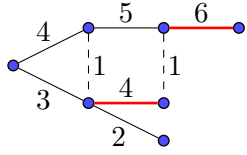

\paragraph{Lower Bounding the QMC Value in Unweighted Graphs} In the unweighted
case we have the advantage that any method for choosing a maximal tree chooses
one of optimal weight, and so (inspired by a method of~\cite{GY21}) we consider
\emph{depth-first search trees}. We will assume the graph is
connected---note that $m$, $W$, and the $\qmcut$ value all sum up over
components, so as long as the lower bound we show is linear in $m$ and $W$ this
will immediately generalize.

In the weighted case, trying to optimize the energy we earned from our tree
meant potentially earning nothing from edges outside the tree, as we had no
control over how they might cross the tree. However, with a DFS tree, we have
the following useful property: for any node in the tree, the subtrees rooted at
its children are disconnected from each other (because otherwise those
connecting edges would have been explored before both subtrees were). This means
we can do the following: choose either the even or odd levels (with level $i$
being edges from depth-$i$ vertices to depth-$(i+1)$ vertices) of the tree, one
of which will contain at least half the edges; call this set of edges $H$.
Now, $H$ consists of disjoint bipartite subgraphs, and no edge outside the tree connects two edges in
the same level of $H$. Thus, as noted above for the weighted case, there is an optimal $\qmcut$ solution for $H$ that still earns
$1/4$ from every edge outside the tree, and from the edges in the unchosen levels.  An optimal classical solution of this kind on $H$ earns a $\qmcut$ value of $1/2$ on each edge in $H$ (by randomly selecting either a fixed assignment that cuts all the edges or the ``bit-flipped'' assignment, independently for each component of $H$).

Now, as the tree contains $W - 1$ edges, merely using the optimal classical solution would only earn us at least $\frac{W - 1}{4} + \frac{m - (W - 1)/2}{4}$,
which is not quite as strong as we want. But each level of the tree is a
disjoint union of stars, and the optimal $\qmcut$ assignment for a star with $d$ leaves earns
$\frac{d+1}{2}$. So we can earn at least $1/2$ more energy, giving us \[
\frac{W - 1}{4} + \frac{m - (W - 1)/2}{4} + \frac{1}{2} > \frac{m}{4} +
\frac{W}{8} 
\]
for a $2$-approximation determined only by $m$ and $W$. We illustrate this
construction in Figure~\ref{fig:dfs_tree}.
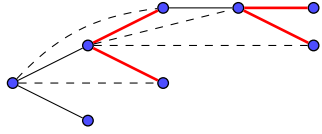
\begin{figure}
\centering
\tikzset{basevertex/.style={shape=circle, line width=0.5, minimum size=4pt,
inner sep=0pt, draw}} \tikzset{defaultvertex/.style={basevertex, fill=blue!70}}
\begin{tikzpicture}[
arrow/.style={single arrow, thick,draw=blue!70,fill=blue!30, minimum
height=10mm}, scale = 1]
\node[style=defaultvertex] (v0) at (0,0) {};
\node[style=defaultvertex] (v1) at (1,0.5) {};
\node[style=defaultvertex] (v2) at (2,1) {};
\node[style=defaultvertex] (v3) at (3,1) {};
\node[style=defaultvertex] (v4) at (4,1) {};
\node[style=defaultvertex] (v5) at (2,0) {};
\node[style=defaultvertex] (v6) at (1,-0.5) {};
\node[style=defaultvertex] (v7) at (4,0.5) {};
\draw (v0) -- (v1);
\draw (v0) -- (v6);
\draw[color=red,line width=1pt] (v1) -- (v2);
\draw[color=red,line width=1pt] (v1) -- (v5);
\draw (v2) -- (v3);
\draw[color=red,line width=1pt] (v3) -- (v4);
\draw[color=red,line width=1pt] (v3) -- (v7);
\draw[dashed] (v0) -- (v5);
\draw[dashed] (v1) -- (v3);
\draw[dashed] (v1) -- (v7);
\draw[dashed] (v0) to [bend left = 20] (v2);
\end{tikzpicture}
\caption{Proving a lower bound for the optimal \qmcut{} value of an unweighted
graph based on a DFS tree (the solid edges in the graph). The heavier half of
the levels in the DFS (colored in red) are given an optimal assignment, and
then every other edge (solid and dashed) earns 1/4. The total energy earned in
this example is $\frac{3}{2} + \frac{3}{2} + \frac{7}{4} = 4.75$.}
\label{fig:dfs_tree}
\end{figure}

\paragraph{Estimating $W$ in the Stream} To obtain an actual algorithm we will
need a $(1 + \varepsilon)$-multiplicative approximation to $m/2 + W/4$. Counting $m$ is
trivial, and in the unweighted case $W$ can be approximated with cardinality
estimation algorithms. So the problem we need to resolve (ideally in $\bO{\log
n}$ space) is estimating \[
\sum_{u \in V} \max_{v \in N(u)}w_{uv}
\]
in the stream. Our approach is to use reservoir sampling to sample edges $e$
with probability proportion to $w_e$, choose an endpoint at random, and then
check whether they are higher-weight than every edge that arrives \emph{after}
them in the stream (since we can't check edges that arrive earlier). If we
defined an estimator that is $1$ whenever this happened and $0$ otherwise, we
would get a contribution of $w_e/2$ for every vertex $u$ and $v \in N(u)$ such
that $w_{uv}$ was a ``scenic viewpoint'', an edge of higher weight than all
subsequent edges incident to $u$.

To correct for this, we also check the weight $w'$ of the highest-weight edge
to arrive incident to $u$ after $uv$ (calling it $0$ if $uv$ is the last edge)
and then subtract $w'/w_{uv}$. This gives us an estimator with expectation \[
W/2m
\]
and constant variance, that we can compute in logarithmic space. So we could
have trouble getting a multiplicative estimate of $W$ if $m \gg W$, but this
isn't a problem---we only want an estimate of $m/2 + W/4$, and so a $\eps
m$-approximation of $W$ suffices. This then gives us our full streaming
algorithm, obtaining a $(2 + \eps)$- and $(5/2 + \eps)$-approximation in the
unweighted and weighted case, respectively, using $\bO{\log n}$ space if $\eps$
is constant.

\section{Preliminaries}
\subsection{Notation} We will generally write $\qubits$ for the number of qubits an
algorithm uses. We will therefore write the state of the algorithm as a density
matrix in $\Cbb^{2^\qubits \times 2^\qubits}$.  However, when considering
\qmcut\ assignments, qubits will typically correspond to vertices in a graph
$(V,E)$, in which case we will use $n$ to denote both $|V|$ and the number of
qubits; we expect this will be clear from context.

In defining \qmcut, we will need to use the Pauli matrices. These are defined as:
\begin{equation*}
\label{eq:paulis}
 I=\begin{bmatrix}
1 & 0 \\
0 & 1
\end{bmatrix},
\,\,\,\,\,\,
X=\begin{bmatrix}
0 & 1 \\
1 & 0
\end{bmatrix},
\,\,\,\,\,\,
Y=\begin{bmatrix}
0 & -i \\
i & 0
\end{bmatrix}, \,\text{and}
\,\,\,\,\,\,
Z=\begin{bmatrix}
1 & 0 \\
0 & -1
\end{bmatrix}.
\end{equation*}
\noindent The notation $\sigma_i$ is used to denote a Pauli matrix $\sigma \in
\{X,Y,Z\}$ acting on qubit $i$: 
\begin{equation*}
\sigma_i = I \otimes I \otimes \ldots \otimes \sigma \otimes \ldots \otimes I \in \mathbb{C}^{2^n \times 2^n}, 
\end{equation*}
where the $\sigma$ occurs at position $i$.  A \emph{Pauli term} or \emph{term}
refers to a tensor product of Pauli operators (e.g., $X_i X_j$).  We say a term
is $k$-local if it is the tensor product of exactly $k$ (non-identity) Pauli
operators (e.g., $I$ is $0$-local, $Y_i$ is $1$-local, and $X_i Z_j$ is
$2$-local). A term is \emph{odd-local} if it is $k$-local for some odd $k$, and
it is \emph{even-local} otherwise. 

The $4^n$ Pauli terms each square to I and are orthonormal under the
Hilbert-Schmidt inner product, i.e., $\frac{1}{2^n}\tra[a^2] =
\frac{1}{2^n}\tra[I] = 1$ and $\frac{1}{2^n}\tra[ab] = 0$ for distinct Pauli
terms $a,b$.  Consequently every Hermitian operator $A$ acting on $n$ qubits
(i.e., $A \in \Cbb^{2^n \times 2^n})$ can be written as a linear combination of
Pauli terms with real coefficients.  

When we write $\norm{.}_p$ for a matrix, we mean the Schatten $p$-norm, defined
as the $p$-norm of the singular values of the matrix (also known as the trace
norm when $p = 1$). The matrix $D$ representing the state of the algorithm will
therefore be positive semi-definite, Hermitian, and satisfy $\norm{D}_1 = 1$.

We will often be using $x \in \bool^n$ to represent a subset of $\brac{n}$. In
a slight abuse of notation, we will sometimes use $x \in \bool^n$ and $x
\subseteq \brac{n}$ interchangeably, when this does not cause ambiguity. We use
the ``weight'' of a set $x$ to refer to $\abs{x}$, usually in the context of a
Fourier coefficient. When dealing with a Fourier coefficient $\wh{f}(x)$
(defined later in this section), we use ``mass'' informally to refer to its
magnitude.

In a graph $G = (V,E)$, we use $N(v)$ to denote the neighborhood of a vertex $v
\in V$, defined as $\set{u \in V : uv \in E}$.

\subsection{Quantum Streaming Algorithms}
As our lower bounds are based on reductions from communication complexity, they
will be valid for the strong definition of quantum streaming, in which the
algorithm maintains $\qubits$ qubits and has a family of quantum channels
$(\Ac_\sigma)_{\sigma \in A}$ where $A$ is the alphabet of all possible updates
that can appear in the stream. On seeing update $\sigma$, the algorithm applies
$\Ac_\sigma$ to its current state. We will therefore not need to concern
ourselves with questions about ancillary qubits---for more discussion on the
subtleties involved here, see section 2.4 of~\cite{NT17}.

\subsection{Quantum and Classical Max-Cut} \label{sec:qmcut-and-mcut}
We define \mcut\ and \qmcut\ as follows.
\begin{definition}[\mcut] \label{def:mcut} Given a graph $G=(V,E)$ with weights $w_{uv} \geq 0$ for $uv \in E$,
find the value of 
\begin{equation*}
\max_{f:V \rightarrow \{-1,1\}} \sum_{uv \in E} w_{uv} \frac{1 - f(u)f(v)}{2}. 
\end{equation*}
An \emph{unweighted} instance is one where $w_{uv} = 1$ for all $uv \in E$.
\end{definition}

In other words, we seek to assign each vertex $u$ to a side of a cut in a graph (a side designated by $f(u) = -1$ or $f(u) = 1$) to maximize the weight of edges crossing the cut.  \mcut\ is an instance of classical $2$-CSP, and one may define a related instance of $2$-LH that is closely related to the anti-ferromagnetic quantum Heisenberg model that has extensively been studied by physicists.

\begin{definition}[\qmcut] \label{def:qmcut} Given a graph $G=(V,E)$ with weights $w_{uv} \geq 0$ for $uv \in E$,
find the maximum eigenvalue (or \emph{energy}) of
\begin{equation*}
Q = \sum_{uv \in E} w_{uv} \frac{I - X_u X_v - Y_u Y_v - Z_u Z_v}{4}.
\end{equation*}
An \emph{unweighted} instance is one where $w_{uv} = 1$ for all $uv \in E$.
\end{definition}

While \mcut\ is \NP-hard, \qmcut\ is \QMA-hard~\cite{Cubitt16Complexity,P15}. Although \qmcut\ is cast as a maximum-eigenvalue problem, we note that $Q \in \mathbb{C}^{2^n \times 2^n}$ is exponentially large in the number of vertices, corresponding to qubits, $n$.  We may also express \mcut\ in such a form, where we seek to find the maximum eigenvalue of
\begin{equation*}
M = \sum_{uv \in E} w_{uv} \frac{I - Z_u Z_v}{2}.
\end{equation*}

To see that this is equivalent to Definition~\ref{def:mcut}, consider
$\ket{\psi} = \ket{s_1\ldots s_n}$, with $s_u \in \{0,1\}$.  Then 
\begin{align}
\nonumber
 \bra{\psi} M \ket{\psi} &=\sum_{uv \in E} w_{uv} \frac{1 - \bra{\psi} Z_u Z_v
 \ket{\psi}}{2}\\
 \nonumber
 &= \sum_{uv \in E} w_{uv} \frac{1 - \bra{s_u} Z \ket{s_u} \bra{s_v} Z
 \ket{s_v}}{2}\\ &= \sum_{uv \in E} w_{uv} \frac{1-(1-2s_u)(1-2s_v)}{2},
\label{eq:mcut-Hamiltonian-value}
\end{align}
and we may take $f(u) = 1-2s_u$ to obtain the desired correspondence.  So maximizing $\bra{\psi} M \ket{\psi}$ over computational basis states $\ket{\psi}$ corresponds to \mcut; however, we observe that $M$ is a diagonal matrix in the computational basis, so that a basis state achieves the maximum eigenvalue.  One way to see that \mcut\ and \qmcut\ are related is that the diagonals of $M$ and $Q$ differ by a factor of two, and so \qmcut\ captures \mcut\ if we restrict ourselves to classical solutions, i.e., basis states.

The quantum nature of \qmcut\ may be better understood by observing that if we look at a graph consisting of a single unweighted edge, we have
\begin{equation}\label{eq:singlet-projector}
Q = \frac{I - X \otimes X - Y \otimes Y - Z \otimes Z}{4} = \ket{\psi^-}\bra{\psi^-},
\end{equation}
where $\ket{\psi^-} = (\ket{01} - \ket{10})/\sqrt{2}$ is the maximally entangled \emph{singlet} state.  Thus the maximum energy of one is obtained by assigning the singlet to the two qubits.  The diagonal of Q is $(\ket{01}\bra{01} + \ket{10}\bra{10})/2$, so a classical solution (basis state) can earn energy at most half by cutting the edge, which corresponds to ``anti-aligning'' or assigning different $\{-1,1\}$-values to the endpoints, in the sense of Definition~\ref{def:mcut}.  A quantum solution can obtain an additional energy of half by taking an ``anti-aligned'' superposition of the two ways of cutting an edge.

\paragraph{Approximation algorithms} We say (following previous work on streaming \mcut\ that an algorithm
$K$-approximates a function $f(G)$ of a graph $G$ if it returns a value in
$\brac{f(G), Kf(G)}$.  In our space-limited setting, we only focus on approximating the objective value of a problem rather than also producing a corresponding feasible solution.

\subsection{Fourier Analysis}
For a function $f$ on the Boolean cube $\bool^n$, the Fourier transform is given by \[
\widehat{f}(S) = \frac{1}{2^n}\sum_{x \in \bool^n} f(x) (-1)^{x \cdot S}
\]
where $S$ ranges over $\bool^n$. Note that this is well-defined if $f$ takes
values in a field $\Fbb$ or vector space over $\Fbb$ such that $\Zbb \subset
\Fbb$. In particular, we will use it when $f$ takes values in $\Cbb$, in the
space of square matrices over $\Cbb$, and in the space of linear operators on
the space of square matrices over $\Cbb$. In the latter case, to reduce
ambiguity about function arguments, we will write \[
\wh{\Ac}_S = \frac{1}{2^n}\sum_{x \in \bool^n} \Ac_x (-1)^{S \cdot x}
\]
when $\Ac_x$ is a family of such linear operators parametrized by $x \in
\bool^n$. As $\sum_{S \in \bool^n} (-1)^{S \cdot x}$ is $2^n$ when $x = 0^n$
and $0$ otherwise, we can also write \[
f(x) = \sum_{S \in \bool^n} \widehat{f}(S) (-1)^{x \cdot S}
\]
and so the Fourier transform is self-inverse up to a scale factor of $2^{-n}$.

We will use the following hypercontractive inequality from~\cite{BARdW08},
which takes the place of an inequality of~\cite{KKL88} that is typically used
in applications of Fourier analysis to classical streaming.
\begin{theorem}[Theorem 1\ in~\cite{BARdW08}]
\label{thm:mathc}
For every $f : \bool^n \rightarrow \Cbb^{m \times m}$ and $1 \le p
\le 2$, \[ \sum_{S \subseteq \brac{n}} (p - 1)^{\abs{S}} \norm{\wh{f}(S)}_p^2
\le \paren*{\frac{1}{2^n}\sum_{x \in \bool^n} \norm{f(x)}_p^p}^{1/p}
\]
where $\norm{\cdot}_1$ is the Schatten (nuclear) 1-norm.
\end{theorem}

We will also use Parseval's identity.
\begin{theorem}[Parseval]
\label{thm:parseval}
For any function $f : \bool^n \rightarrow \Rbb$, \[
\frac{1}{2^n}\sum_{x \in \bool^n} f(x)^2 = \sum_{S \subseteq \brac{n}} \wh{f}(S)^2\text{.}
\]
\end{theorem}
\paragraph{Useful Lemmas} The convolution lemmas and applications of
hypercontractivity in this section will be extensions of standard results in
Fourier analysis, but we will need slightly stronger or more general versions,
so we prove them here.

We will need an extension of the convolution theorem for Fourier transforms to
the matrix-valued case.
\begin{lemma}
\label{lm:matconv}
Let $f$, $g$ be functions on $\bool^n$ such that the range of $f$ is contained
in either $\Cbb$ or $\Cbb^{a \times b}$, and the range of $g$ is contained in
either $\Cbb$ or $\Cbb^{b \times c}$, where $a, b, c \in \Nbb$. Then \[
\wh{fg}(S) = \sum_{T \in \bool^n} \wh{f}(T)\wh{g}(T \oplus S)
\]
for all $S \in \bool^n$.
\end{lemma}
\begin{proof}
\begin{align*}
\wh{fg}(S) &= \E[x]{f(x)g(x)(-1)^{S \cdot x}}\\
&= \E[x]{\sum_{T \in \bool^n}\wh{f}(T)(-1)^{T \cdot x}g(x)(-1)^{S \cdot x}}\\
&= \sum_{T \in \bool^n}\wh{f}(T) \E[x]{g(x)(-1)^{(T \oplus S) \cdot x}}\\
&= \sum_{T \in \bool^n}\wh{f}(T)\wh{g}(T \oplus S)
\end{align*}
\end{proof}
We will use Theorem~\ref{thm:mathc} to derive two lemmas that bound the Fourier
mass of functions that output density matrices (or, more generally, matrices
with bounded trace norm), one for lower-degree coefficients and one for all
coefficients. Both will be based on the following corollary of
Theorem~\ref{thm:mathc}. Our proof is along the lines of that of Lemma~10\
in~\cite{AD21}, but in a slightly different setting.
\begin{lemma}
\label{lm:matkklcor}
For any $f : \bool^n \rightarrow \Cbb^{2^\qubits \times 2^\qubits}$ such that
$\norm{f(x)}_1 \le 1$ for all $x$, for any $0 \le \delta \le 1$, \[
\sum_{S \subseteq \brac{n}} \delta^{\abs{S}} \norm{\wh{f}(S)}_1^2 \le
2^{2\delta \qubits}\text{.}
\]
\end{lemma}
\begin{proof}
Setting $p = 1 + \delta$, we obtain
\begin{align*}
\sum_{S \subseteq \brac{n}} \delta^{\abs{S}} \norm{\wh{f}(S)}_{1 + \delta}^2 &\le
\paren*{\frac{1}{2^n} \sum_{x \in \bool^n} \norm{f(x)}_{1 + \delta}^{1 + \delta}}^{\frac{1}{1 + \delta}}\\
&\le \paren*{\frac{1}{2^n} \sum_{x \in \bool^n} \norm{f(x)}_{1}^{1 +
\delta}}^{\frac{1}{1 + \delta}}\\
& = 1
\end{align*}
Next, note that for any $q$, $2^{-\qubits}\norm{\wh{f}(S)}_q^q$ can be viewed as
$\E{\abs{X}^q}$, where $X$ is uniformly chosen from the singular values of
$\wh{f}(S)$. So by Jensen's inequality, we obtain
\begin{align*}
\norm{\wh{f}(S)}_{1 + \delta} &= \paren*{2^\qubits \E{\abs{X}^{1+\delta}}}^{\frac{1}{1 +
\delta}}\\
&\ge 2^{\frac{\qubits}{1 + \delta}} \E{\abs{X}}\\
&= 2^{\frac{\qubits}{1 + \delta} - \qubits} \norm{\wh{f}(S)}_1\\
&\ge 2^{-\delta \qubits}\norm{\wh{f}(S)}_1
\end{align*}
and so the corollary follows.
\end{proof}
We can now bound lower order terms of the Fourier expansion with a careful
choice of $\delta$.
\begin{lemma}
\label{lm:lodegbound}
For any $f : \bool^n \rightarrow \Cbb^{2^\qubits \times 2^\qubits}$ such that
$\norm{f(x)}_1 \le 1$ for all $x$, and for any $k \in \brac{\qubits}$,
\begin{align*}
\sum_{\substack{S \subseteq \brac{n}:\\ \abs{S} = k}} \norm{\wh{f}(S)}_1^2 &\le
\paren*{\frac{\bO{\qubits}}{k}}^k\\
\sum_{\substack{S \subseteq \brac{n}:\\ \abs{S} = k}} \norm{\wh{f}(S)}_1 &\le
\paren*{\frac{\bO{\sqrt{\qubits n}}}{k}}^k
\end{align*}
\end{lemma}
\begin{proof}
Set $\delta = k/\qubits$ in Lemma~\ref{lm:matkklcor}. Then we get
\begin{align*}
\sum_{\substack{S \subseteq \brac{n}:\\ \abs{S} = k}} \norm{\wh{f}(S)}_1^2 &\le
\paren*{\frac{\qubits}{k}}^k\sum_{S \subseteq \brac{n}}
\paren*{\frac{k}{\qubits}}^{\abs{S}} \norm{\wh{f}(S)}_1^2\\
&\le  \paren*{\frac{\qubits}{k}}^k 2^{2k}\\
&= \paren*{\frac{\bO{\qubits}}{k}}^k\text{.}
\end{align*}
The second line then follows from applying Cauchy-Schwarz, as for any
non-negative-valued function $g$ on subsets of $\brac{n}$, \[
\sum_{S \subseteq \brac{n}} g(S) \le \sqrt{\sum_{S \subseteq \brac{n}} g(S)} \cdot
\sqrt{\sum_{S \subseteq \brac{n}} 1} \le  \sqrt{\sum_{S \subseteq \brac{n}} g(S)}
\cdot \sqrt{\binom{n}{k}}
\]
\end{proof}
We can also obtain a general bound with a less careful choice of $\delta$.
\begin{lemma}
\label{lm:hidegbound}
For any $f : \bool^n \rightarrow \Cbb^{2^\qubits \times 2^\qubits}$ such that
$\norm{f(x)}_1 \le 1$ for all $x$,
\begin{align*}
\sum_{S \subseteq \brac{n}} \norm{\wh{f}(S)}_1^2 &\le 2^{2\qubits}\\
\sum_{S \subseteq \brac{n}} \norm{\wh{f}(S)}_1 &\le \paren*{\frac{O(n)}{k}}^{k/2}
\end{align*}
\end{lemma}
\begin{proof}
Set $\delta = 1$ in Lemma~\ref{lm:matkklcor}. The second line then follows from
Cauchy-Schwarz as in the previous proof.
\end{proof}
For scalar-valued functions, a stronger version of the above (for the 2-norm)
is given by Parseval's identity\footnote{This can be proved for matrix-valued
functions too, but we will not need it.} (Theorem~\ref{thm:parseval}).  This
will allow us to characterize the Fourier coefficients of a function that
checks whether its input satisfies a collection of linear constraints.
\begin{lemma}
\label{lm:linearcon}
Let $M \in \bool^{k \times n}$ be a matrix over $\Zbb_2$. Let $y \in \bool^k$,
and define $q : \bool^n \rightarrow \bool$ by \[
q(x) = \begin{cases}
1 & \mbox{if $Mx = y$}\\
0 & \mbox{otherwise.}
\end{cases}
\]
Then for every $s \in \bool^k$, \[
\widehat{q}(M\tr s) = \frac{\abs{q^{-1}(1)}}{2^n}(-1)^{s \cdot y}
\]
and every other Fourier coefficient of $q$ is $0$.
\end{lemma}
\begin{proof}
We first prove the part of this theorem pertaining to coefficients of the form $M\tr s$. We have 
\begin{align*}
\widehat{q}(M\tr s) &= \frac{1}{2^n}\sum_{x \in \bool^n} q(x) (-1)^{x\cdot (M\tr s)}\\
&= \frac{1}{2^n}\sum_{x \in \bool^n} q(x) (-1)^{x\tr M\tr s}\\
&= \frac{1}{2^n}\sum_{x \in \bool^n} q(x) (-1)^{y\tr s}\\
&= \frac{\abs{q^{-1}(1)}}{2^n}(-1)^{s \cdot y}
\end{align*}
where the third line makes use of the fact that $q(x) = 1$ iff $Mx = y$. We
will now use Parseval's identity (Theorem~\ref{thm:parseval}) to show that every
other coefficient must be 0. Now, if $Mx = y$ has no solutions $x$, every
coefficient will be $0$ and the theorem will hold trivially. 

So suppose it has at least one solution. Then by the Rank-Nullity theorem and
the fact we are working over $\Zbb_2$, the number of solutions
$\abs{q^{-1}(1)}$ is $2^{n-r}$, where $r$ is the dimension of the row space of
$M$. Furthermore, $2^r$ is the number of distinct coefficients $M^ts$. So
\begin{align*}
\sum_{\substack{z \in \bool^n :\\ \exists s \in \bool^k, T = M\tr s}}\wh{q}(z)^2
&= 2^r \frac{\abs{q^{-1}(1)}^2}{2^{2n}}\\
&= \frac{\abs{q^{-1}(1)}}{2^{n}}\\
&= \frac{1}{2^n}\sum_{x \in \bool^n}q(x)^2
\end{align*}
which by Parseval is the sum of the square of \emph{every} Fourier coefficient
of $q$, so as $q$ is real-valued its other Fourier coefficients must all be
$0$.
\end{proof}

As well as being used directly, this will also allow us to restrict the set of
non-zero Fourier coefficients for any function that only depends on a linear
function of its inputs.
\begin{lemma}
\label{lm:zerocoeffs}
Let $f$ be a real or complex (matrix or scalar)-valued function on $\bool^n$,
and let $M$ be a matrix over $\Zbb_2$ and $g$ another function such that $f(x)
= g(Mx)$ for all $x \in \bool^n$. Then for every $S \in \bool^n$, \[
\wh{f}(S) = 0
\]
unless $S$ is $M\tr s$ for some $s \in \bool^k$.
\end{lemma}
\begin{proof}
We can write \[
f(x) = \sum_{y \in \bool^k} q_y(x) g(y)
\]
where $q_y$ is the indicator function of $Mx = y$, as in
Lemma~\ref{lm:linearcon}. So then, if $S$ is not $M\tr s$ for some $s\in
\bool^k$, $\wh{q}_y(S) = 0$ for all $y$ by Lemma~\ref{lm:linearcon} and so
$\wh{f}(S) = 0$.
\end{proof}

Finally, we will need a generalization of the convolution theorem to arbitrary
linear operators. For a family of linear operators $(\Ac_x)_{x \in \bool^n}$ we
will write, by analogy to the Fourier transform, \[
\wh{\Ac}_s = \frac{1}{2^n}\sum_{x \in \bool^n} (-1)^{x \cdot s} \Ac_x \text{.}
\]
\begin{lemma}
\label{lm:lopconv}
Let $(\Ac_x)_{x \in \bool^n} : \Lc_1 \rightarrow \Lc_2$ be a family of linear
operators, $f : \bool^n \rightarrow \Lc_1$ any function, and $g : \bool^n
\rightarrow \Lc_2$ be given by \[
g(x) = \Ac_xf(x)\text{.}
\]
Then \[
\wh{g}(S) = \sum_{T \in \bool^n} \wh{\Ac}_T \wh{f}(S \oplus T)
\]
for all $S \in \bool^n$.
\end{lemma}
\begin{proof}
\begin{align*}
 \sum_{T \in \bool^n} \wh{\Ac}_T \wh{f}(S \oplus T) &=
\frac{1}{2^{2n}} \sum_{T,x,y \in \bool^n}(-1)^{x\cdot T + y \cdot (S \oplus
T)}\Ac_xf(y)\\
&= \frac{1}{2^{2n}} \sum_{T,x,y \in \bool^n}(-1)^{(x \oplus y) \cdot T + y
\cdot S }\Ac_xf(y)\\
&= \frac{1}{2^{n}} \sum_{x,y \in \bool^n}(-1)^{y
\cdot S }\Ac_xf(y) \Ibb(x = y)\\
&= \frac{1}{2^{n}} \sum_{x \in \bool^n}(-1)^{x \cdot S}\Ac_xf(x)\\
&= \wh{g}(S)
\end{align*}
\end{proof}

\section{Communication Problem}
We will use the Distributional Implicit Hidden Partition (\dihp) problem
of~\cite{KK19}. 

In a \dihp$(n, \alpha, \players)$ instance, for $n, \players \in \Nbb$ and
$\alpha \in (0, 1/2)$ such that $\alpha n \in \Nbb$, $\players$ players indexed
by $t \in \brac{\players}$ each receive the incidence matrix of a matching $M_t
\in \bool^{\alpha n \times n}$ along with bit labels $\elabels_t \in
\bool^{\alpha n}$.

The bit labels are private to their respective players, while each player $t$
knows $M_s$ for all $s \le t$. The players will communicate sequentially, from
player $i$ to player $i+1$, and the objective is for player $T$ to determine
which of the following two distributions the inputs were drawn from:
\begin{description}
\item[\yes] Draw $\vlabels \sim \Uni{\bool^n}$. For each $t \in \brac{T}$, set
$\elabels_t = M_t\vlabels$.
\item[\no] For each $t \in \brac{T}$, draw $w_t \sim \Uni{\bool^{\alpha n}}$.
\end{description}

\section{Reducing Quantum and Classical Max-Cut to \dihp{}}
In this section we will prove that any \mcut\ or QMC algorithm gives a \dihp\
protocol. 
\begin{restatable}{theorem}{reduction}
\label{thm:reduction}
There exists a constant $C > 0$ such that, for any $\varepsilon,\delta \in (0,
1\rbrack$ and $n$ a multiple of $\ceil{C/\eps} \cdot \ceil{C^3/\eps^2}$, the
following holds: Suppose there is a streaming algorithm that returns a $(2 -
\varepsilon)$-approximation to the \mcut\ size of an $n$-vertex graph, or a $(2
- \eps)$-approximation to its QMC value, in $S$ space with probability $1 -
\delta$. Then \dihp$(n, 1/\ceil{C/\eps}, \ceil{C^3/\eps^2})$ has a protocol
using $S + 2\log n$ space that succeeds with probability $1 - \delta -
2^{-n}$.
\end{restatable}
An immediate consequence of this is that, if we can show that solving \dihp$(n,
\alpha, T)$ with probability $2/3$ requires $\bOm{n}$ space for all
sufficiently small constants $\alpha$, $T$, then any algorithm giving a $(2 -
\varepsilon)$-approximation to \mcut\ or QMC with $2/3$ probability for some
constant $\varepsilon$ requires $\bOm{n}$ space.

Our reduction will be based on a graph encoding of \dihp. We can generate a
graph from a \dihp\ instance by, for each matching $M_t$ and the $i\nth$ edge in
that matching, adding the edge to the graph iff $(w_t)_i = 1$ and the edge is
not in $Ms$ for any $s \le t$. Note that as player $t$ knows $(M_s)_{s \le t}$,
they can construct this graph in the stream, as their own input and $(M_s)_{s
\le t}$ will suffice to determine which edges to insert.

We write $\Gc^Y$ for the distribution on graphs generated thus conditioned on
the \dihp\ instance being a \yes\ case, and $\Gc^N$ for the distribution
conditioned on a \no\ case. Let $\Gb$ denote a draw from one of these
distributions, and $\mb$ the number of edges in $\Gb$. We will show, in both
cases, a $(2 - \varepsilon)$-separation between the cut value of a draw from
$\Gc^Y$ and $\Gc^N$ (as a fraction of $\mb$), with high probability.
\begin{lemma}
\label{lm:bipcut}
Let $\Gb$ be drawn from $\Gc^Y$. Then the $\mcut$ value of $\Gb$ is $\mb$ and
the QMC value is at least $\mb/2$.
\end{lemma}
\begin{proof}
$\Gb$ will always be bipartite, as every edge will cross the bipartition given
by $\vlabels$. Therefore, the $\mcut$ value will be $\mb$, and as any classical
cut always earns at least half its value for QMC, the QMC value is at least
$\mb/2$.
\end{proof}

To lower-bound the cut value, we will start with classical $\mcut$ and then
show that this implies the corresponding QMC bound.
\begin{lemma}
\label{lm:randmcut}
Let $\Gb$ be drawn from $\Gc^N$, and let $n \ge T$. Then, with probability $1
- e^{n(1 -\bOm{\alpha^2 T})}$, the $\mcut$ value of $\Gb$ is less than
$\mb/(2 - \bO{1/n + \alpha})$.
\end{lemma}
\begin{proof}
We start by proving that $\mb$ is at least $\alpha T n/8$ with high
probability. Consider the probability that the $i\nth$ edge of the $t\nth$
matching is included in $\Gb$, conditioned on on $(M_s)_{s = 1}^{t-1}$ and
edges $1, \dots, i-1$ of $M_t$. Then, the $i\nth$ edge of $M_t$ is uniformly
distributed among those edges not incident to any of the $i - 1$ edges from
$M_t$ conditioned on, and it will be included with probability $1/2$ if it is
none of the edges in $(M_s)_{s=1}^{t-1}$. So it is included with probability at
least \[
\frac{1}{2}\cdot\frac{\binom{n}{2} - (i - 1)n - \alpha (t - 1) n}{\binom{n}{2}}
\le \frac{n^2 - n - 2\alpha n^2 - 2\alpha nT}{2n^2} \le \frac{1}{2} - 1/n -
\alpha - \alpha T/n \le \frac{1}{2} - 1/n - 2\alpha
\] 
as $n \ge T$. We may assume $1/n - 2\alpha \le 1/4$, as otherwise the lemma
follows trivially by choosing the constants in $\bOm{\cdot}$ and $\bO{\cdot}$
to be large enough.  So the distribution of $\mb$ stochastically dominates
$\Bi{n\alpha T}{1/4}$ and so by the Chernoff bounds, \[
\mb \ge \frac{\alpha n T}{8} 
\]
with probability at least \[
1 - e^{-\bOm{n\alpha T}}\text{.}
\]
Now, conditioning on any value of $\mb \ge \frac{\alpha n T}{8}$, consider any
cut of the $\brac{n}$ vertices. Now, consider the probability that any one of
those $\mb$ edges crosses the cut. For the $i\nth$ edge of $M_t$, even if we
condition on the entire state of the game except for $(M_t)_i$, and
additionally condition on the edge being included in the graph, its location is
uniformly distributed over at least $\binom{n}{2} - 2\alpha n^2 - \alpha n T
\ge n^2(1 - 1/n - 5\alpha)/2$ possible locations (as the other edges from the
same player's $\alpha n$-matching exclude no more than $2\alpha n^2$
possibilities, $2n$ for each of them, while the edges from other players'
matchings can exclude at most $1$ edge each, for a total of no more than
$\alpha n T$). So the probability it will cross the cut is at most \[
\frac{1}{2(1 - 1/n - 5\alpha)} 
\]
as at most $n^2/4$ of the \emph{possible} locations for the edge will cross the
cut. As this holds conditioning on any value of the rest of the input, it will
also hold conditioning on any part of it, and so the number of edges crossing
the cut is stochastically dominated by $\Bi{\mb}{\frac{1}{2 - 1/n - 5\alpha}}$,
and so by again applying the Chernoff bounds and assuming that $1 - 1/n -
6\alpha \ge 1/2$ (as otherwise the lemma again holds trivially), it is less
than \[
\frac{\mb}{2(1 - 1/n - 6\alpha)} = \frac{\mb}{2 - 1/n - 5\alpha} +
\bOm{\alpha m}
\] with probability at least \[
1 - e^{-\bOm{\alpha^2 m}} = 1 - e^{-\bOm{\alpha^3 n^2 T}} =1 - e^{-\bOm{\alpha^2 n T}}
\]
conditioned on any $\mb \ge n \alpha T/8$ and using the fact that $\alpha n \in
\Nbb$. The lemma follows by taking a union bound over all $2^n$ possible cuts.
\end{proof}

This upper bounds the \mcut{} value of a \no{} instance. To extend this bound
to \qmcut{}, we use the fact that the natural semidefinite programming (SDP)
relaxation of \mcut\ also gives an SDP relaxation of QMC, and this relaxation
becomes tight as the max-cut value tends towards $1/2$. Specifically, we use
the vector program
\begin{equation} \label{eq:MC-vector-program}
\max_{f : V \rightarrow S^{n-1}} \sum_{uv \in E} - \langle f(u), f(v) \rangle
\end{equation}
for a graph $G = (V,E)$. The above vector program is equivalent to the standard Goemans-Williamson vector program
\[
\max_{f : V \rightarrow S^{n-1}} \sum_{uv \in E} \frac{1- \langle f(u), f(v) \rangle}{2},
\]
and may be solved as an SDP~\cite{GoemansWilliamson95}.
Thus if the optimal value of \eqref{eq:MC-vector-program} is $K$, the $\mcut$ value of $G$ is at most $(\mb + K)/2$.  A related vector program also gives an upper bound for QMC, so that its value is at most $(\mb + 3K)/4$~\cite{GharibianParekh19}.

Conversely, we also have that if the $\mcut$ value of $G$ is at most $n/2 +
\eps$, then $K = O(\eps)$, by~\cite{CW04}.
This allows us to prove the following lemma:
\begin{lemma}
\label{lm:randqmc}
Let $\Gb$ be drawn from $\Gc^N$, and let $n \ge T$. Then, with probability $1
- e^{1-\bO{\alpha^2 n T}}$, the QMC value of $\Gb$ is less than
$\mb/(4 - \bO{1/n + \alpha})$.
\end{lemma}
\begin{proof}
By Lemma~\ref{lm:randmcut}, with this probability the $\mcut$ value is at least
$\mb/(2 - \bO{1/n + \alpha})$, and so the SDP value is $\bO{(1/n +
\alpha)\mb}$. This in turn implies that the QMC value is less than $\mb/(4 -
\bO{1/n + \alpha})$.
\end{proof}

We may now prove Theorem~\ref{thm:reduction}.
\reduction*
\begin{proof}
The protocol will be as follows:
\begin{itemize}
\item Each player in turn gives every edge in their matching $M_t$ with bit
label 1 to the algorithm, unless the edge appears in a previous player's
matching $M_s$. They keep track of $\mb$, the number of edges input this way.  
\item The final player uses the algorithm to approximate \mcut/QMC, returning \yes\
if the returned value is at least $\mb/(2 - \varepsilon)$ (for $\mcut$) or at
least $\mb/(4 - \varepsilon)$ (for QMC) and \no\ otherwise.
\end{itemize}
This will cost $S$ space to maintain the algorithm, and $2 \log n$ space to maintain
the counter.

The player's input is now a draw from $\Gc^Y$ conditioned on the problem being
in a $\yes$ instance, or from $\Gc^N$ conditioned on it being in a $\no$
instance. Lemma~\ref{lm:bipcut} guarantees that a draw from $\Gc^Y$ will have $\mcut$ value $\mb$ and QMC value at least $\mb/2$.

Provided $C$ is chosen to be large enough, Lemmas~\ref{lm:randmcut} and~\ref{lm:randqmc} say that a draw from $\Gc^N$ will have $\mcut$ value less than $\mb/(2 - \eps)$ and QMC value less than $\mb/(4 - \eps)$ with probability at least \[
1 - 2^{-n(1 - \bO{\ceil{C/\eps}^{-2}\cdot \ceil{C^3/\eps^2}})} = 1 - 2^{-n}
\]
and so as the algorithm succeeds with probability $1 - \delta$, the protocol
will return the correct answer with probability at least $1 - \delta -
2^{-n}$.
\end{proof}

\section{Quantum Lower Bound for \dihp{}}
\subsection{Reduction to Fourier Analysis}
\label{sec:foured}
Let $\qubits$ be the number of qubits each player may send to the next. For
simplicity we will think of the final player, $T$, sending an $\qubits$-qubit
message (consisting of either $\yes$ or $\no$ with another $\qubits-1$ bits
fixed at $\ket{0}$). We will think of the first player as receiving a message
from ``player 0'' that is fixed as $\ket{0}^\qubits$. {We write $g_t :
\bool^{\alpha \players n} : \rightarrow \Cbb^{2^\qubits \times 2^\qubits}$ for
player $t$'s message as a function of the edge labels of the first $t$
matchings, and $f_t : \bool^n : \rightarrow \Cbb^{2^\qubits \times 2^\qubits}$
for $f_t(x) = g_t\paren*{\paren*{M_sx}_{s=1}^t}$. Note that both functions
depend on the first $t$ matchings $(M_s)_{s=1}^t$. Note also that as these
messages are density matrices, they in particular have trace norm 1.

Now, we define \[
\phi_t = \E[x \sim \uni{\bool^n}, y \sim \uni{\bool^{n \alpha (T - t)}}]{g_T\paren*{(M_sx)_{s=1}^t,y}}
\]
so for any given set of matchings $(M_t)_{t=1}\tr $, the total variation
difference between the final player's output in a $\yes$ and a $\no$ case is
given by \[
\norm{\phi_T - \phi_0}_1.
\]
\begin{lemma}
\label{lm:phibound}
\[
\norm{\phi_T - \phi_0}_1 \le \sum_{t=1}^{T-1}\sum_{s \in \bool^{\alpha
n}\setminus \set{\emptyset}}\norm*{\wh{f}_t(M_{t+1}\tr s)}_1
\]
\end{lemma}
\begin{proof}
We first note that $\phi_1$ and $\phi_0$ are identical, as $M_1x$ is uniformly
distributed. It will therefore suffice to prove that, for all $t \in
\brac{T-1}$, \[
\norm{\phi_{t+1} - \phi_{t}}_1 \le \sum_{s \in \bool^{\alpha
n}\setminus \set{\emptyset}}\norm*{\wh{f}_t(M_{t+1}\tr s)}_1
\]
and the theorem will then follow by the triangle inequality. Now,
we note that $\phi_{t+1}$ and $\phi_{t}$ are given by applying the same quantum
channel to \[
\E[x \sim \uni{\bool^n}]{g_{t+1}\paren*{(M_sx)_{s=1}^{t+1}}}
\]
and \[
\E[x \sim \uni{\bool^n}, y \sim \uni{\bool^{\alpha n}}]{g_{t+1}\paren*{(M_sx)_{s=1}^{t},y}}
\]
respectively. Therefore, it will suffice to bound the trace norm of \[
\E[x \sim \uni{\bool^n}, y \sim \uni{\bool^{\alpha
n}}]{g_{t+1}\paren*{(M_sx)_{s=1}^{t+1}} - g_{t+1}\paren*{(M_sx)_{s=1}^{t},y}}
\] which we may rewrite as \[
\E[x \sim \uni{\bool^n}, y \sim \uni{\bool^{\alpha
n}}]{\Ac_{M_{t+1}x}\paren*{f_{t}(x) -
\Ac_{y}\paren*{f_t(x)}}}
\]
where $\Ac_y$ is the quantum channel that player $t+1$ applies to the message
received from player $t$ to generate the message sent onwards, if player
$(t+1)$'s bit labels are $y$ (this channel may also depend on the matchings
$(M_s)_{s=1}^{t+1}$). 

Writing $q_y$ for the indicator function on $y = M_{t+1}x$, this in turn equals
\[
\E[x \sim \uni{\bool^n}, y \sim \uni{\bool^{\alpha
n}}]{2^{\alpha n}q_y(x)\Ac_{y}\paren*{f_t(x)} -
\Ac_{y}\paren*{f_t(x)}}
\]
which by the linearity of $\Ac_y$ can be rewritten as \[
\E[y \sim \uni{\bool^{\alpha n}}]{\Ac_y\paren*{2^{\alpha n}\wh{q_yf_t}(\emptyset) - \wh{f}_t(\emptyset)}}\text{.}
\]
Now, by Lemmas~\ref{lm:matconv}~and~\ref{lm:linearcon},
\begin{align*}
\wh{q_yf_t}(\emptyset) &= \sum_{S \in \bool^n} \wh{q}_y(S)\wh{f}_t(S)\\
&= \sum_{s \in \bool^{\alpha n}} \frac{\abs{q^{-1}_y(1)}}{2^n} \wh{f}_t(M_{t+1}\tr s)(-1)^{s\cdot y}\\
&= 2^{-\alpha n} \sum_{s \in \bool^{\alpha n}} \wh{f}_t(M_{t+1}\tr s) (-1)^{s\cdot y}
\end{align*}
as a $2^{-\alpha n}$ fraction of strings $x$ satisfy $M_{t+1}x = y$. Putting
this together, we get
\begin{align*}
\norm{\phi_{t+1} - \phi_{t}}_1 &\le \norm*{\E[y \sim \uni{\bool^{\alpha
n}}]{\Ac_y\paren*{\sum_{s \in \bool^{\alpha n}\setminus
\set{\emptyset}}\wh{f}_t(M_{t+1}\tr s)(-1)^{s\cdot y}}}}_1\\
&\le \E[y \sim \uni{\bool^{\alpha n}}]{\norm*{\Ac_y\paren*{\sum_{s \in \bool^{\alpha
n}\setminus \set{\emptyset}}\wh{f}_t(M_{t+1}\tr s) (-1)^{s\cdot y}}}_1}\\
&= \E[y \sim \uni{\bool^{\alpha n}}]{\norm*{\sum_{s \in \bool^{\alpha
n}\setminus \set{\emptyset}}\wh{f}_t(M_{t+1}\tr s) (-1)^{s\cdot y}}_1}\\
&\le \sum_{s \in \bool^{\alpha
n}\setminus \set{\emptyset}}\norm*{\wh{f}_t(M_{t+1}\tr s)}_1
\end{align*}
completing the proof.
\end{proof}

\subsection{Evolution of Fourier Coefficients}
In this section we will show that, in expectation over the matchings, the
even-degree Fourier coefficients of $f$ remain small enough in magnitude for
Lemma~\ref{lm:phibound} to give a useful bound. 
\begin{restatable}{lemma}{evdegbound}
\label{lm:evdegbound}
There exists constants $C > 1, D > 1$ such that, for all $\ell \in
\brac{n/2}$ and $t \in \brac{T}$, if $\qubits \le n/D^t$ and $\alpha \le
1/D$, \[ \sum_{\substack{S \in \bool^n\\ \abs{S} = 2\ell}}
\E[(M_s)_{s=1}^t]{\norm{\widehat{f}_t(S)}_1} \le
\paren*{\frac{C^tn}{\ell}\cdot\max\set*{\frac{\qubits}{\ell},
1}}^{\ell/2}\text{.} \]
\end{restatable}
We write $\Ac_x^t$ for the quantum channel player $t$ applies to the state
received from player $t-1$, if the input player $t$ sees is $(M_t, M_tx)$ (the
family $(\Ac_x)_{x \in \bool^n}$ will therefore depend on $M_t$ and nothing
else).  This means that $f_t(x) = \Ac_x^tf_{t-1}(x)$ and so we may apply
Lemma~\ref{lm:lopconv} to get \[
\wh{f}_t(S) = \sum_{U \in \bool^n} \wh{\Ac}_U \wh{f}_{t-1}(U \oplus S)
\]
for all $S \in \bool^n$. Now we will show that we only need to care about the
coefficients $U$ that correspond to sets of edges in the matching $M_t$.
\begin{lemma}
For all $S \in \bool^n$, $\wh{\Ac}^t_S = 0$ unless $S = M_t\tr s$ for some $s
\in \bool^{\alpha n}$.
\end{lemma}
\begin{proof}
Suppose $S$ is not $M_t\tr s$ for any $s$. Then for any $N \in \Cbb^{2^\qubits
\times 2^\qubits}$, define \[
h(x) = \Ac_x^tN\text{.}
\]
As $\Ac^t_x$ depends only on $M_tx$, we may apply Lemma~\ref{lm:zerocoeffs} to
get that $\wh{h}(S) = 0$ and therefore $\wh{\Ac}^t_SN = 0$. As this applies for
all $N$, $\wh{\Ac}^t_S = 0$.
\end{proof}
Applying this gives us
\[
\wh{f}_t(S) = \sum_{s \in \bool^{\alpha n}} \wh{\Ac}^t_{M_t\tr s}
\wh{f}_{t-1}(M\tr s \oplus S)
\] and summing over $S$ of weight $2\ell_2$ for any $\ell_2 \in \brac{n/2}$, we get
\begin{align*}
\sum_{\substack{S \in \bool^n\\\abs{S} = 2\ell_2}} \norm{\wh{f}_t(S)}_1 &\le
\sum_{\substack{S \in \bool^n\\ \abs{S} = 2\ell_2}} \sum_{s \in
\bool^{\alpha n}} \norm{\wh{\Ac}^t_{M_t\tr s}\wh{f}_{t-1}(M\tr s \oplus S)}_1\\
&= \sum_{S \in \bool^n} \sum_{\substack{s \in \bool^{\alpha n}\\ \abs{M\tr s
\oplus S} = 2\ell_2}} \norm{\wh{\Ac}^t_{M_t\tr s}\wh{f}_{t-1}(S)}_1\text{.}
\end{align*}
Now we can write the total mass on weight-$\ell_2$ coefficients of $\wh{f}_{t}$
in terms of the contributions from each weight of coefficients of
$\wh{f}_{t-1}$. Grouping $s \in \bool^{\alpha n}$ by the (always even) weight
of $M\tr s \oplus S$, we get, 
\[
\sum_{\substack{S \in \bool^n\\\abs{S} = 2\ell_2}} \norm{\wh{f}_t(S)}_1 \le
\sum_{\ell_1 = 0}^{\floor{n/2}} \tau_t(\ell_1,\ell_2)
\]
where $\tau_t(\ell_1, \ell_2)$ is a bound on the amount of ``mass transferred''
between coefficients of weight $\ell_1$ and $\ell_2$ at step $t$, given by
\[
\tau_t(\ell_1, \ell_2) \coloneqq \sum_{\substack{S \in \bool^n\\ \abs{S} =
2\ell_1}}\sum_{\substack{s \in \bool^{\alpha n}\\ \abs{M_t\tr s \oplus S} =
2\ell_2}} \norm{\wh{\Ac}^t_{M\tr s}\wh{f}_{t-1}(S)}_1\text{.}
\]
We now split $s \in \bool^{\alpha n}$ into smaller groups, based on parameters
$a, b, c \in \Nbb$, where $a$ is the number of edges from $M\tr s$ contained in
$S$, $b$ is the number with one endpoint in $S$ and $c$ is the number outside
of $S$ (and so requiring that $\abs{M\tr s \oplus S} = 2\ell_2$ is equivalent
to requiring that $\ell_1 + c - a = \ell_2$). Let $I_S(y)$ be the indicator
variable on $M_t\tr y$ being contained in $S$, and $B_S(y)$ on the $i\nth$
column of $M_t\tr$ having intersection exactly $1$ with $S$ whenever $y_i = 1$.
Then we have \[
\tau_t(\ell_1, \ell_2) \le \sum_{\substack{a,b,c \in \Nbb\\ \ell_1 + c - a = \ell_2}}\tau(\ell_1,a,b,c)
\]
for \[
\tau_t(\ell_1,a,b,c) \coloneqq \sum_{\substack{S \in \bool^{n}\\ \abs{S} =
2\ell_1}} \sum_{\substack{u \in \bool^{\alpha n}\\\abs{u} =
a}}\sum_{\substack{v \in \bool^{\alpha n}\\ \abs{v} = b}}
I_S(u)B_S(v)\sum_{\substack{w \in \bool^{\alpha n}\\ \abs{w} = c}}
\norm{\wh{\Ac}^t_{M_t\tr(u \oplus v \oplus w)}\wh{f}_{t-1}(S)}_1\text{.}
\]
We are now ready to bound this ``mass transfer'' in terms of the of
weight-$\ell_1$ coefficients of $\widehat{f}_{t-1}$, in expectation over $M_t$.
This will hold for any values of $(M_s)_{s=1}^{t-1}$ and $(x_s)_{s=1}^t$ (note
that $M_t$ is distributed independently of all these variables).
\begin{lemma}
\label{lm:transfer}
\[
\E[M_t]{\tau_t(\ell_1,a,b,c)} \le 2^{O(\ell_1 + c)} \cdot (\alpha \cdot a/n)^a \cdot
(n/c)^{c/2} \cdot \max\set{(\qubits/c)^{c/2}, 1} \cdot \sum\limits_{\substack{S \in
\bool^{n}\\ \abs{S} = 2\ell_1}} \norm{\wh{f}_{t-1}(S)}_1 
\]
\end{lemma}
\begin{proof}
We start by bounding the innermost sum in $\tau_t$, regardless of $u$, $v$, or
$M_t$. Let $\Bc^t_y$ denote the channel applied by player $t$ on seeing $y$, so
that $\Ac_x = \Bc_{M_t x}$ for all $x$. For any fixed $S \in \bool^n$, let $g:
\bool^{\alpha n}\rightarrow
\Cbb^{2^\qubits \times 2^\qubits}$ be given by \[
g(y) = (-1)^{(u \oplus v)\cdot y}\Bc^t_{y}\wh{f}_{t-1}(S)/\norm{\wh{f}_{t-1}(S)}_1\text{.}
\]
Then for all $w \in \bool^{\alpha n}$,
\begin{align*}
\wh{g}(w) &= \frac{1}{2^{\alpha n}}\sum_{y \in \bool^{\alpha n}} (-1)^{(u \oplus v \oplus w)\cdot
y}\Bc^t_{y}\wh{f}_{t-1}(S)/\norm{\wh{f}_{t-1}(S)}_1\\
&= \frac{1}{2^{\alpha n}}\sum_{y \in \bool^{\alpha
n}}\frac{1}{2^{(1-\alpha)n}}\sum_{\substack{x \in \bool^n\\ M_tx = y}} (-1)^{(u
\oplus v \oplus w)\cdot M_tx} \Ac^t_x \wh{f}_{t-1}(S)/\norm{\wh{f}_{t-1}(S)}_1\\
&= \frac{1}{2^n}\sum_{x \in \bool^n} (-1)^{M_t\tr (u \oplus v \oplus w) \cdot
x} \Ac^t_x \wh{f}_{t-1}(S)/\norm{\wh{f}_{t-1}(S)}_1\\
&= \wh{\Ac}^t_{M_t\tr(u \oplus v \oplus w)}\wh{f}_{t-1}(S) / \norm{\wh{f}_{t-1}(S)}_1
\end{align*}
and as applying a quantum channel can only shrink the trace norm,
$\norm{g(x)}_1 \le 1$ for all $x$. So we may apply Lemmas~\ref{lm:lodegbound}
and~\ref{lm:hidegbound} to obtain 
\begin{align*}
\frac{1}{\norm{\wh{f}_{t-1}(S)}_1}\sum_{\substack{w \in \bool^{\alpha n}\\ \abs{w} = c}}
\norm{\wh{\Ac}^t_{M_t\tr(u \oplus v \oplus w)}\wh{f}_{t-1}(S)}_1 &=
\sum_{\substack{w \in \bool^{\alpha n}\\ \abs{w} = c}}
\norm{\wh{g}(w)}_1\\
&\le \begin{cases}
\paren*{\frac{\sqrt{\bO{\qubits \cdot \alpha n}}}{c}}^c & \mbox{$c \le \qubits$}\\
\paren*{\frac{\bO{\alpha n}}{c}}^{c/2} &\mbox{$c > \qubits$}
\end{cases}\\
&\le \begin{cases}
2^{O(c)}\cdot c^{-c} \cdot \paren*{\qubits n}^{c/2} & \mbox{$c \le \qubits$}\\
2^{O(c)} \cdot c^{-c/2} \cdot n^{c/2} & \mbox{$c > \qubits$}
\end{cases}\\
&= 2^{O(c)} \cdot (n/c)^{c/2} \cdot \max\set{(\qubits/c)^{c/2},1}\text{.}
\end{align*}
With the inner sum bounded independently of $u$, $v$, and $M_t$, it will then suffice to have a bound on \[
\E[M_t] {\sum_{\substack{u \in \bool^{\alpha n}\\\abs{u} =
a}}\sum_{\substack{v \in \bool^{\alpha n}\\ \abs{v} = b}} I_S(u)B_S(v)}
\]
that holds for all $S$.

First note that, there are $\binom{n}{2a}$ possible sets of vertices that can
be matched by a size-$a$ subset of the matching $M_t$, and $\binom{\ell_1}{2a}$
of them are contained in $S$. So $I_S(u)$ is $1$ with probability
$\binom{\ell_1}{2a}/\binom{n}{2a}$. 

Then, for any $u$, conditioned on $I_S(u) = 1$, any $v$ not disjoint from $u$
has probability zero of having $B_S(v) = 1$, as no edge can have both endpoints
in $S$ while also having exactly one endpoint. Meanwhile, for a size-$b$
disjoint $v$, there are $2^b$ ways to choose one endpoint from each edge in the
subset of $M_t$ indexed by $v$, and then a $\binom{b}{\ell_1 - a}/\binom{b}{n -
a}$ probability that these endpoints are all in $S$, so $\E{B_S(v) | I_S(u) =
1} \le 2^b \binom{b}{\ell_1 - a}/\binom{b}{n - a}$.

Therefore, and assuming $a, b \le \ell_1$ and $\ell_1 \le n/2$ as
otherwise the lemma would hold trivially,
\begin{align*}
\E[M_t] {\sum_{\substack{u \in \bool^{\alpha n}\\\abs{u} = a}}\sum_{\substack{v
\in \bool^{\alpha n}\\ \abs{v} = b}} I_S(u)B_S(v)} & \le \sum_{\substack{u \in
\bool^{\alpha n}\\\abs{u} = a}}\sum_{\substack{v \in \bool^{\alpha n}\\ \abs{v}
= b}} \frac{\binom{\ell_1}{2a}}{\binom{n}{2a}} \cdot 2^b \cdot \frac{\binom{\ell_1 -
a}{b}}{\binom{n - a}{b}}\\
&\le \sum_{\substack{u \in \bool^{\alpha n}\\\abs{u} = a}}\sum_{\substack{v \in
\bool^{\alpha n}\\ \abs{v} = b}} \frac{2^{\ell_1}}{(n/2a)^{2a}} \cdot
2^b \cdot \frac{2^{\ell_1}}{((n - a)/b)^b}\\
&=  \sum_{\substack{u \in \bool^{\alpha n}\\\abs{u} = a}}\sum_{\substack{v \in
\bool^{\alpha n}\\ \abs{v} = b}} 2^{O(\ell_1)} \cdot a^{2a} b^b \cdot
n^{-2a - b}\\
&= \binom{\alpha n}{a} \binom{\alpha n}{b} 2^{O(\ell_1)} \cdot a^{2a} b^b \cdot
n^{-2a - b}\\
&\le 2^{O(\ell_1)} \cdot (\alpha \cdot a)^a \cdot n^{-a}
\end{align*}
and so the lemma follows.
\end{proof}

This will give us the inductive step we need to prove
Lemma~\ref{lm:evdegbound}. We will need one further small lemma about integers
for our proof.
\begin{lemma}
\label{lm:factprod}
For any $x, y, z \in \Nbb_{>0}$ and $x \le y + z$, \[
\frac{x^x}{y^y \cdot z^z} \le \frac{2^{O(y + z)}}{(y + z - x)^{y + z - x}}
\]
\end{lemma}
\begin{proof}
Without loss of generality we will assume that $y \le z$.
\begin{align*}
\frac{x^x}{y^y \cdot z^z} &= 2^z \frac{x^x}{y^y \cdot (2z)^z}\\
&\le  2^z \frac{x^x}{(y + z)!}\\
&\le 2^{O(y + z)} \frac{x^x}{(y + z)^{y + z}}\\
&\le 2^{O(y + z)} \frac{1}{(y + z)^{y + z - x}}\\
&\le \frac{2^{O(y + z)}}{(y + z - x)^{y + z - x}}
\end{align*}
\end{proof}
\evdegbound*
\begin{proof}
We proceed by induction on $t$, proving the slightly stronger version that
includes $t = 0$. For $t = 0$, the result holds trivially because $f_0$ is
constant and so $\norm{\wh{f}_0(S)}_1 \le 1$ if $S = \emptyset$ and $0$
otherwise.

Now, suppose there is a constant $C > 0$ such that the lemma holds for $t -
1$. We will show that, if $C$ is chosen to be a large enough constant, \[
\sum_{\substack{S \in \bool^n\\ \abs{S} = 2\ell}}
\E[(M_s)_{s=1}^t]{\norm{\widehat{f}_{t}(S)}_1} \le \paren*{\frac{C^{t}n}{\ell}
\cdot \max\set*{\frac{\qubits}{\ell},1}}^{\ell/2}
\]
for any given $\ell \in \Nbb$. We will split this into two parts, the low
degree case ($\ell \le \qubits$) and the high degree case ($\ell > \qubits$).

\paragraph{Low Degree Case} We have
\begin{equation}
\sum_{\substack{S \in \bool^n\\ \abs{S} = 2\ell}} \norm{\widehat{f}_{t}(S)}_1
\le \sum_{\ell' =1}^\qubits \tau_t(\ell',\ell) + \sum_{\ell' =
\qubits+1}^{\floor{n/2}} \tau_t(\ell',\ell)\text{.} \label{eq:taudec}
\end{equation}
We will start by using the inductive hypothesis and Lemma~\ref{lm:transfer} to
bound the first sum in~\eqref{eq:taudec}. Recall that we can bound
$\tau_t(\ell',\ell)$ as a sum over $\tau_t(\ell', a, b, c)$ such that $\ell =
\ell' + c - a$, and note that $\tau_t(\ell', a, b, c)$ can only be non-zero 0
if $2a + b \le 2\ell'$. For any $a, b, c$ satisfying both criteria (which as
$\ell'$ and $\ell$ are both at most $\qubits$, also imply $c < \qubits$, we
have, by Lemma~\ref{lm:transfer} and our inductive hypothesis,
\begin{align*}
	 \E[(M_s)_{s=1}^t]{\tau_t(\ell', a, b, c)} &\le 2^{O(\ell' + c)} \cdot
	 (a/n)^a \cdot (\sqrt{\qubits n}/c)^c \cdot
	 \paren*{\frac{\sqrt{C^{t-1}\qubits n}}{\ell'}}^{\ell'}\\ 
	 &= 2^{O(\ell' + c)} \cdot \frac{a^a}{c^c {\ell'}^{\ell'}} \cdot (\qubits
	 n)^{(\ell' + c - a)/2} \cdot (\qubits /n)^{a/2} \cdot C^{\frac{\ell'(t -
	 1)}{2}}\\ 
	 &\le  2^{O(\ell' + c)} \cdot \paren*{\frac{\sqrt{\qubits n}}{\ell}}^\ell
	 \cdot (\qubits /n)^{\max\set{(\ell' - \ell)/2, 0}} \cdot C^{\frac{\ell'(t -
	 1)}{2}}
\end{align*}
Now, using the fact that $a, b \le \ell'$ and so there are at most $\ell'$
possibilities for each (and since $\ell - \ell' = c - a$, at most $\ell'$
possibilities for $c$), we have \[
 \E[(M_s)_{s=1}^t]{\tau_t(\ell',\ell)} \le  2^{O(\ell' + c)} \cdot
\paren*{\frac{\sqrt{\qubits n}}{\ell}}^\ell \cdot (\qubits /n)^{\max\set{(\ell'
- \ell)/2, 0}} \cdot C^{\ell'(t - 1)/2}
\]
and so, using the fact that $c = \ell - \ell' + a \le \ell$, 
\begin{align*}
\sum_{\ell' = 1}^\qubits   \E[(M_s)_{s=1}^t]{\tau_1 (\ell', \ell)} &=
\sum_{\ell'=1}^\ell  \E[(M_s)_{s=1}^t]{\tau_1 (\ell', \ell)} + \sum_{\ell' =
\ell + 1}^\qubits   \E[(M_s)_{s=1}^t]{\tau_1 (\ell', \ell)}\\
&\le \paren*{\frac{\sqrt{\qubits n}}{\ell}}^\ell \cdot \paren*{\sum_{\ell' = 1}^\ell
2^{O(\ell + \ell')} \cdot C^{\frac{\ell'(t-1)}{2}} + \sum_{\ell'= \ell+1}^\qubits  2^{O(\ell + \ell')} \cdot
C^{\frac{\ell'(t-1)}{2}} (\qubits /n)^{\frac{(\ell' - \ell)}{2}}}\\
&\le \paren*{\frac{\sqrt{\qubits n}}{\ell}}^\ell \cdot \paren*{\sum_{\ell' = 1}^\ell
2^{O(\ell + \ell')} \cdot C^{\frac{\ell'(t-1)}{2}} + \sum_{\ell'= \ell+1}^\qubits  2^{O(\ell + \ell')} \cdot
C^{\frac{\ell'(t-1)}{2}} D^{\frac{-t(\ell' - \ell)}{2}}}\\
&\le \paren*{\frac{\sqrt{\qubits n}}{\ell}}^\ell \cdot 2^{O(\ell)} \cdot C^{\frac{\ell(t - 1)}{2}}\\
&\le \frac{1}{2}\paren*{\frac{\sqrt{C^t\qubits n}}{\ell}}^\ell
\end{align*}
provided that $C$ is chosen to be a sufficiently large constant and $D$ is
chosen to be sufficiently larger than $C$.

Now, we use Lemma~\ref{lm:transfer}, the inductive hypothesis, and
Lemma~\ref{lm:hidegbound} to bound the second sum in~\eqref{eq:taudec}. Putting
these together, we have for $a,b,c$ such that $\ell = \ell' + c - a$, $\ell' >
\qubits  \ge \ell$, and $\tau_t(\ell',a,b,c)$ is non-zero (which in particular
implies $2a + b \le 2\ell'$),
\begin{align*}
 \E[(M_s)_{s=1}^t]{\tau_t(\ell', a, b , c)} &\le 2^{O(\ell' +  c)} \cdot (\alpha \cdot a/n)^a
\cdot (n/c)^{c/2} \cdot \sum_{\substack{S \in \bool^n\\ \abs{S} = 2\ell'}}
\norm{\widehat{f}_{t-1}(S)}_1\\
&\le 2^{O(\ell')} \cdot \alpha^a \cdot (a/n)^{a/2} \cdot (n/c)^{c/2}
\cdot \min\set*{\paren*{\frac{C^{t-1}n}{\ell'}}^{\ell'/2},
\paren*{\frac{n}{\ell'}}^{\ell'}} \cdot (a/n)^{a/2}\\\
&\le 2^{O(\ell')} \cdot \alpha^{\ell'-\ell} \cdot (a/n)^{a/2} \cdot (n/c)^{c/2}
\cdot \paren*{\frac{C^{t-1}n}{\ell'}}^{\ell/2} \cdot
\paren*{\frac{n}{\ell'}}^{(\ell' - \ell)} \cdot (a/n)^{(\ell' - \ell)/2}\\\
&\le 2^{O(\ell')} \cdot \alpha^{\ell'-\ell} \cdot
\frac{a^{a/2}}{c^{c/2}{\ell'}^{\ell'/2}} \cdot n^{(\ell' + c - a)/2} \cdot
C^{(t - 1)\ell/2} \cdot (a/\ell')^{(\ell'-\ell)/2}\\
&\le 2^{O(\ell')} \cdot \alpha^{\ell' - \ell} \cdot
 \paren*{\frac{n}{\ell' + c - a}}^{(\ell' + c - a)/2} \cdot C^{(t - 1)\ell/2}\\
&\le 2^{O(\ell')} \cdot \alpha^{\ell' - \ell} \cdot
 \paren*{\frac{n}{\ell}}^{\ell/2} \cdot C^{(t - 1)\ell/2}\\
&\le 2^{O(\ell')} \cdot \alpha^{\ell' - \ell} \cdot
 \paren*{\frac{\sqrt{\qubits n}}{\ell}}^{\ell} \cdot C^{(t - 1)\ell/2}\\
&\le 2^{O(\ell')} \cdot (1/D)^{\ell' - \ell}\cdot
 \paren*{\frac{\sqrt{C^{t-1}\qubits n}}{\ell}}^{\ell} \text{.}
\end{align*}
making use of Lemma~\ref{lm:factprod}. Again using the fact that there are no
more than $\ell'$ choices for each of $a, b, c$ that give non-zero $\tau_t$, we
get \[
 \E[(M_s)_{s=1}^t]{\tau(\ell', \ell)} \le 2^{O(\ell')} \cdot (1/D)^{\ell' -
 \ell}\cdot \paren*{\frac{\sqrt{C^{t-1}\qubits n}}{\ell}}^{\ell} 
\]
and so summing over all $\ell' > m$ gives us \[
\sum_{\ell' = \qubits  + 1}^{\floor{n/2}}  \E[(M_s)_{s=1}^t]{\tau(\ell', \ell)}  \le
\frac{1}{2}\paren*{\frac{\sqrt{C^t\qubits n}}{\ell}}^\ell
\]
provided $C$ and $D$ are chosen to be sufficiently large. Adding our bounds on
the first and second sum in~\eqref{eq:taudec} completes the inductive step for
$\ell \le \qubits $. 

\paragraph{High Degree Case} For this case we will use the fact that $(x/y)^y
\le \binom{x}{y} \le 2^x$ when $y \in \brac{x}$. Let $\ell \in (\qubits ,
n/2\rbrack \cap \Zbb)$. For any $\ell', a, b, c$ such that $a, b \le \ell'$ and
$\ell' + c - a = \ell$ we get, by Lemma~\ref{lm:transfer}, our inductive
hypothesis, and Lemma~\ref{lm:hidegbound},
\begin{align*}
\tau_t(\ell',a,b,c)  &\le 2^{O(\ell' + c)} \cdot (\alpha \cdot a/n)^a \cdot
(n/c)^{c/2} \cdot \max\set{(\qubits /c)^{c/2}, 1} \\
&\phantom{\le} \cdot \min\set*{\paren*{\frac{C^{t-1}n}{\ell'} \cdot \max\set*{\frac{\qubits }{\ell'},
1}}^{\ell'/2}, \paren*{\frac{n}{\ell'}}^{\ell'}}\\
& \le 2^{O(\ell' + c)} \cdot \alpha^a \cdot (a/n)^{a/2} \cdot
(n/c)^{c/2} \cdot 2^\beta \\
&\phantom{\le} \cdot \paren*{\frac{C^{t-1}n}{\ell'}}^{\min\set{\ell/2, \ell'/2}} \cdot 2^{\ell'} \cdot
(n/\ell')^{\max\set{(\ell'- \ell), 0}}\cdot (a/n)^{a/2} \\
&\le 2^{O(\ell + \ell')} \cdot \frac{a^{a/2}{\ell'}^{\ell'/2}}{c^{c/2}}
\cdot n^{(\ell' - a + c)/2} \cdot C^{(t-1)\ell/2} \cdot \alpha^{\max\set{\ell' - \ell,0}} \\
&\phantom{\le} \cdot C^{(t-1)\min\set*{(\ell' - \ell)/2, 0}}\cdot
(n/\ell')^{\max\set*{(\ell - \ell')/2, 0}} \cdot (a/n)^{\max\set*{(\ell -
\ell')/2,0}}\\&
\le 2^{O(\ell + \ell')} \cdot \paren*{\frac{C^{t-1} n}{\ell}}^{\ell/2} \cdot
\max\set*{1/C, 1/D}^{\abs{\ell' - \ell}}
\end{align*}
and once again we use the fact that there are no more than $\ell'$ choices for
$a, b, c$ to obtain \[
 \E[(M_s)_{s=1}^t]{\tau(\ell', \ell)} \le 2^{O(\ell' + c)} \cdot
 \paren*{\frac{C^{t-1} n}{\ell}}^{\ell/2} \cdot \max\set*{1/C, 1/D}^{\abs{\ell'
 - \ell}}\text{.}
 \]
We may then sum over all $\ell'$ to obtain\[
\sum_{\ell' = \qubits  + 1}^{\floor{n/2}}  \E[(M_s)_{s=1}^t]{\tau(\ell', \ell)}  \le
\paren*{\frac{\sqrt{C^tn}}{\ell}}^{\ell/2}
\]
provided $C$ and $D$ are chosen to be sufficiently large.
\end{proof}

\subsection{Completing the Bound}
\begin{theorem}
\label{thm:dihplb}
For any constant $\delta > 0$, there exists a constant $C > 0$ such that, for
any $\alpha < 1/C$, $n, T \in \Nbb$, $\dihp(n,\alpha,T)$ requires at least
$n/C^T$ communication to solve with probability $1/2 + \delta$ in the one-way
quantum communication model.
\end{theorem}
\begin{proof}
Let $C', D'$ correspond to the constants $C,D$ from Lemma~\ref{lm:evdegbound}.
We will choose $C$ to be larger than $\max{(C')^6,D', 2}$. Consider a protocol
using $\qubits \le n/C^T$ communication (with $C$ to be specified later). We
will start by assuming that $\qubits \ge \log T$, then show how to remove this
assumption. Using the notation of Section~\ref{sec:foured}, the probability
that the final player outputs the correct answer for any given matchings
$(M_t)_{t=1}^T$ is bounded by \[
\frac{1}{2} + \frac{1}{2} \norm{\phi_T - \phi_0}_1
\]
and so (applying Lemmas~\ref{lm:phibound} and~\ref{lm:evdegbound}, and
Lemma~\ref{lm:hidegbound} for higher-degree Fourier terms), the success
probability is bounded by $1/2$ plus 
\begin{align*}
\E[(M_t)_{t=1}^T]{\norm{\phi_T - \phi_0}_1} & \le \sum_{t=1}^{T-1}\sum_{s \in \bool^{\alpha
n}\setminus \set{\emptyset}} \E[(M_t)_{t=1}^T]{\norm*{\wh{f}_t(M_{t+1}\tr s)}_1}\\
&= \sum_{t=1}^{T-1} \sum_{\ell = 1}^{\alpha n} \sum_{\substack{S \in \bool^n\\
\abs{S} = 2\ell}} \sum_{s \in \bool^{\alpha n}}
\E[(M_s)_{s=1}^{t}]{\norm{\wh{f}(S)}_1 \cdot \prob[M_{t+1}]{M_{t+1}\tr s = S}}\\
&\le \sum_{t=1}^{T-1} \sum_{\ell = 1}^{\alpha n} \sum_{\substack{S \in
\bool^n\\ \abs{S} = 2\ell}} \E[(M_s)_{s=1}^{t}]{\norm{\wh{f}(S)}_1} \cdot
\binom{\alpha n}{\ell}/\binom{n}{2\ell} \\
&\le \sum_{t=1}^{T-1} \sum_{\ell = 1}^{\alpha n} \paren*{\frac{O(\alpha \cdot
\ell)}{n}}^\ell \sum_{\substack{S \in \bool^n\\ \abs{S} = 2\ell}}
\E[(M_s)_{s=1}^{t}]{\norm{\wh{f}(S)}_1}\\
&\le  \sum_{t=1}^{T-1} \paren*{\sum_{\ell = 1}^{\qubits}
\paren*{\frac{O(\ell)}{n}}^\ell \cdot \paren*{\frac{C^{t/3} \sqrt{\qubits
n}}{\ell}}^\ell +  \sum_{\ell = \qubits +1}^{\alpha n} \paren*{\frac{O(\alpha
\cdot \ell)}{n}}^\ell \cdot \paren*{\frac{\bO{n}}{\ell}}^{\ell}}\\
&\le \sum_{t=1}^{T-1} \paren*{\sum_{\ell = 1}^{\qubits }C^{-t\ell/6} +
\sum_{\ell =
\qubits +1}^{\alpha n} (1/C)^\ell}\\
&= \bO{C^{-1/6} + T/C^\qubits }\\
& < \delta
\end{align*}
provided $C$ is chosen to be large enough. So the protocol cannot achieve $2/3$ success probability.

We now have a constant $C > 0$ such that no protocol using $\qubits \in (\log
T, n/C^T)$ communication can succeed. We will now show that no protocol using
$\qubits < n/(2C)^T$ communicating can succeed, which will suffice to prove the
lemma. If $(\log T, n/C^T)$ contains at least one integer $\qubits'$, we can
``pad'' any $\qubits < \log T$ qubit protocol to use $\qubits'$ qubits by
adding $\qubits' - \qubits$ qubits that are never used, without affecting the
success probability. So as no $\qubits'$ qubit protocol can succeed, neither
can any $\qubits$ qubit protocol. Conversely, if the interval $(\log T, n/C^T)$
contains no integers, $n/C^T < \log T + 1$, and so $n/(2C)^T < (\log T + 1)/C^T
< 1$ as $C > 2$. So then the result holds trivially.
\end{proof}

Our main result now follows as a corollary.
\maxcutlb*
\begin{proof}
By setting the constant in $\bO{\cdot}$ large enough that $2^{\bO{1/\eps^2}}
\ge \ceil{C'/\eps} \cdot \ceil{{C'}^3/\eps^2}$, we may reduce to \dihp{} via
Theorem~\ref{thm:reduction} and then apply Theorem~\ref{thm:dihplb} to conclude
that any algorithm for $(2-\varepsilon)$-approximating \mcut\ with probability
$2/3$ requires at least $n/C^{\ceil{{C'}^3/\varepsilon^2}} - 2\log n$ bits,
where we use $C'$ for the constant $C$ in Theorem~\ref{thm:reduction}. Then
there is a constant $D$ such that $n/D^{1/\varepsilon^2} <
\max\set*{n/C^{\ceil{{C'}^3/\varepsilon^2}} - 2\log n, 1}$, and so the result
follows (as clearly any algorithm must use at least 1 qubit of storage).
\end{proof}

\section{Space Complexity Upper Bounds for Quantum Max-Cut}
\label{sec:qmcapprox}
Unlike in the classical case, where the trivial algorithm offers the best
result possible in $\lO{n}$ space, the trivial algorithm for \qmcut\ offers
only a 4-approximation. In this section we show how to attain a
$(2+\eps)$-approximation for unweighted \qmcut\ and a $(5/2 +
\eps)$-approximation to weighted \qmcut. We will show this by proving (in
Lemmas~\ref{lm:wmclb} and~\ref{lm:umclb}) that these approximations can be
derived from two quantities, the sum of the edge weights, and the sum of the
max-weight edge at each vertex (in the unweighted case, the number of edges and
non-isolated vertices, respectively). 

The first of these quantities can be calculated exactly with a single
counter---we show that the second can also be calculated in small space
(Lemma~\ref{lm:wstream}), giving us as a consequence the following result.
\qmcub*

\subsection{Approximating \qmcut\ from Simple Graph Parameters}
For a (weighted) graph $G = (V,E, w)$ we will write $m = \sum_{e \in E} w_e$
and $W = \sum_{u \in V} \max_{v \in N(u)} w_{uv}$. In the unweighted case we
will use the same definitions with every weight considered to be $1$ (so $m$ is
the number of edges and $W$ the number of non-isolated vertices). We have the
following upper bound for the \qmcut\ of a graph.
\begin{lemma}[Lemma 5, \cite{Anshu2020}] \label{lem:qmcut-bound}
The Quantum Max-Cut value of $G$ is at most $m/2 + W/4$.
\end{lemma}
\begin{remark}The \qmcut\ terms employed in~\cite{Anshu2020} are twice the ones we employ.  Their Lemma 5 gives a bound for stars, and taking half the sum of this bound applied to the star around each vertex in $G$ yields the bound stated in the above lemma.
\end{remark}

We will use this to obtain approximations to $\qmcut$ in terms of $m$ and $W$.
We will use the
following facts about $\qmcut$:
\begin{lemma}[\qmcut\ facts]\leavevmode
\label{lem:qmc-facts}
\begin{enumerate}[(a)]
\item \label{item:mc-qmc-values} 
The \qmcut\ value is at least half the classical max-cut value.\label{qmc:class}
\item \label{item:singlet-value}
For any given edge $uv$, we can earn energy $w_{uv}$ for it by
assigning a singlet to its two vertices.\label{qmc:singlet}
\item \label{item:star-assignment}
For a vertex with $d$ incident weight-1 edges, there is an assignment to
it and its neighbors that earns total energy $\frac{d+1}{2}$.\label{qmc:star}
\item \label{item:disjoint-set-assignment}
For any set of assignments to
disjoint subgraphs $(H_i)_{i \in I}$, there is an assignment that preserves the
energy of these assignments on their subgraphs while earning $1/4$ for every
edge \emph{between} a pair of subgraphs.\label{qmc:quarter}
\end{enumerate}
\end{lemma}
\begin{proof} For the following we let $Q_{uv} = (I - X_u X_v - Y_u Y_v - Z_u Z_v)/4$ be a \qmcut\ term acting on qubits $u$ and $v$.  

For \ref{item:mc-qmc-values}, if $\ket{\psi} = \ket{s_1\ldots s_n}$, with $s_u \in \{0,1\}$, then $\bra{\psi} X_u X_v \ket{\psi} = \bra{\psi} Y_u Y_v \ket{\psi} = 0$.  Thus, following Equation~\eqref{eq:mcut-Hamiltonian-value}, 
\begin{equation*}
\bra{\psi} Q_{uv} \ket{\psi} = \frac{1 - \bra{\psi} Z_u Z_v \ket{\psi}}{4} = \frac{1}{2} \frac{1-(1-2s_u)(1-2s_v)}{2}, 
\end{equation*}
and if we interpret each $s_u$ as assigning $u$ to a side of a cut, then $\bra{\psi} Q_{uv} \ket{\psi}$ is precisely half the value this cut earns on edge $uv$.

For \ref{item:singlet-value} we recall from
Equation~\eqref{eq:singlet-projector} that assigning a
singlet\footnote{Technically we mean taking the $n$-qubit state $\rho =
Q_{uv}$.} to qubits $u$ and $v$ earns a maximum weighted energy of $w_{uv}$.

For \ref{item:star-assignment}, consider a vertex $u$ with $d$ incident edges,
$\{uv\}_{v \in N}$ ($i$ may have additional incident edges to vertices outside
$N$). The proof of Lemma 5 in \cite{Anshu2020} demonstrates that there is a
quantum state $\ket{\psi}$ on $u \cup N$ such that $\bra{\psi} \sum_{uv : v \in
N} Q_{uv} \ket{\psi} = \frac{d+1}{2}$.  (Moreover, $\psi$ resides in the
single-excitation subspace,
$\text{span}\{\ket{10\ldots0},\ket{01\ldots0},\ldots,\ket{00\ldots1}\}$, and
its coefficients in this subspace correspond to the coefficients of a maximum
eigenvector of a $(d+1) \times (d+1)$ matrix, a graph Laplacian in fact.)

For \ref{item:disjoint-set-assignment} suppose we have assigned mixed states $\rho_i$ to disjoint subgraphs $H_i$ of our graph $G$.  In other words we assign the state $\rho = \otimes_i \rho_i$ to the qubits in $\cup_i H_i$.  We show that we may replace each $\rho_i$ with a state $\rho'_i$ so that
\begin{enumerate}
\item[(i)] $\rho'_i$ contains no 1-local terms, and 
\item[(ii)] $\tra[Q_{uv} \rho'_i] = \tra[Q_{uv} \rho_i]$, for all $uv \in H_i$. 
\end{enumerate}
Property (ii) ensures that the energy earned on edges within each $H_i$ are preserved.  Property (i) implies that $\rho' = \otimes_i \rho'_i$ cannot contain any terms of the form $X_u X_v$, $Y_u Y_v$, or $Z_u Z_v$ for $u$ and $v$ in different $H_i$; hence, $\tra[X_u X_v \rho'] =  \tra[Y_u Y_v \rho'] = \tra[Z_u Z_v \rho'] = 0$, and consequently $\tra[Q_{uv} \rho'] = 1/4$, for such $uv$.  

We construct $\rho'_i$ from $\rho_i$ as follows.  For a Hermitian operator $A$, let $A^-$ be the operator obtained from $A$ by negating the coefficient of each odd-local term.  We have $\tra[\rho^-_i] = \tra[\rho_i] = 1$, and we next show that $\rho^-_i \succeq 0$ since $\rho_i \succeq 0$.  Proceed by contrapositive and assume there is a $\ket{\psi}$ such that $\bra{\psi}  \rho^-_i \ket{\psi} < 0$.  Set $P = \ket{\psi}\bra{\psi}$.  We have 
\begin{equation}\label{eq:rho'-positivity}
\tra[\rho_i (P^-)^2] = \tra[\rho_i (P^2)^-] = \tra[\rho^-_i P^2] = \tra[\rho^-_i P] = \bra{\psi} \rho^- \ket{\psi} < 0,
\end{equation}
where the third equality follows because $P^2 = P$.  The first two equalities hold because for Hermitian $A$ and $B$ acting on $n$ qubits,
\begin{itemize}
\item $(A^-)^2 = (A^2)^-$: the only terms appearing in $B^2$ are those of the form $ab$ where $a,b$ are commuting terms of $B$; consequently $ab$ is odd-local if and only if exactly one of $a,b$ is.
\item $\tra[AB^-] = \tra[A^-B]$: we have  $\frac{1}{2^n}\tra[AB^-] = \langle u, v^- \rangle = \langle u^-, v \rangle = \frac{1}{2^n}\tra[A^-B]$, where $u,v$ are Bloch vectors\footnote{A \emph{Bloch} vector of a Hermitian operator $A$ acting on $n$ qubits is the vector of the $4^n$ real coefficients of the Pauli terms of $A$.} of $A,B$ respectively, and $w^-$ negates the entries of $w$ corresponding to odd-local terms.   
\end{itemize}
Since $\rho \succeq 0$, $\tra[\rho_i A^2]$ must be nonnegative for any Hermitian $A$, and Equation~\eqref{eq:rho'-positivity} demonstrates that $\rho^-_i \succeq 0$.  Finally we get the desired state $\rho'_i = (\rho_i + \rho^-_i)/2$, which has no odd-local terms, and the even-local terms remain the same as $\rho_i$.  The latter ensures that Property (ii) holds.
\end{proof}

\begin{lemma}
\label{lm:wmclb}
The Quantum Max-Cut value of $G$ is at least $m/5 + W/10$.
\end{lemma}
\begin{proof}
We use a modified version of an argument in the proof of Theorem~3\
in~\cite{Anshu2020}. If $m \ge 2W$, we can use the fact that the \qmcut\ value
is at least $m/4$ (by~\ref{qmc:class} and considering a random assignment) to
get a lower bound of \[
m/4 = m/5 + m/20 \ge m/5 + W/10\text{.}
\]
So suppose $m < 2W$. Now order the edges of $G$ by weight (with arbitrary but
consistent tiebreakers when edges have the same weight). Let $H$ be the
subgraph formed by choosing the maximum weight edge incident to each vertex.
Note that $H$ is a forest, as for any cycle, each edge must have greater weight
than at least one of the next or the previous edge in the cycle. But maintaining
this through the cycle would eventually require that the first edge in the
cycle had greater weight than itself.

Let $\Mc$ denote all the edges in $H$ that were ``chosen'' by (had the greatest
weight of any edge incident to) both their endpoints. This is then a matching.
Let $\Fc$ denote the rest of $H$. Let $M$ and $F$ denote the total weight of
each, so $2M + F = W$. As $\Mc$ is a matching we may use \ref{qmc:singlet} and
\ref{qmc:quarter} to choose an assignment to it that earns $w_e$ on each edge
$e$ in the matching, plus $w_e/4$ for every other edge $e$ in the graph, for a
total \qmcut\ value of $M + m/4$. 

As $\Fc \cup \Mc$ is a forest, and therefore has a value $F + M$ classical cut,
\ref{qmc:class} tells us that its \qmcut\ value is at least $F/2 + M/2$. So
substituting in $F = W - 2M$, the \qmcut\ value is at least \[
\max\set{M + m/4, W/2 - M/2}
\]
and so as $M \ge 0$, and $W/2 > m/4$, this is minimized when $M + m/4 = W/2 - M/2$, so when $3M/2 = W/2 - m/4$. We can therefore conclude that the \qmcut\ value is at least
\begin{align*}
m/12 + W/3 &= m/5 - 7m/60 + W/3\\
&\ge m/5 - 7W/30 + W/3\\
&= m/5 + W/10
\end{align*}
completing the proof.
\end{proof}

\begin{lemma}
\label{lm:umclb}
Let $G$ be unweighted. Then the Quantum Max-Cut value of $G$ is at least $m/4 + W/8$.
\end{lemma}
\begin{proof}
We will use a technique inspired by a method of~\cite{GY21} for proving lower
bounds for weighted \mcut.  In particular we take advantage of the fact that, for any
vertex in a DFS, the subtrees rooted at its children will be disconnected from
each other. (If there were an edge between $u$ and $v$ in two different
subtrees, and $u$ was explored before $v$, then edge $uv$ would have been followed by
the DFS, and $v$ would have been included in the same subtree as $u$.) 

We will assume $G$ is connected---by property~\ref{qmc:quarter}, the $\qmcut$
value of a graph is at least the sum of the $\qmcut$ value of its components,
so as the value of $m$ and $W$ for $G$ is the sum of the values for its
components, this will suffice to prove the result for general $G$.

Now let $T$ be a DFS tree for $G$, started from an arbitrary vertex $v$. $T$
will contain $W - 1$ edges (as $W$ is just the vertex count for an unweighted
graph). For all $i \in \Nbb$, let $T_i$ denote the set of edges in $T$ between
a vertex at depth $i$ in the tree and one at depth $i+1$, so $T =
\bigcup_{i=1}^W T_i$. Note that each $T_i$ is a union of stars (formed by the
depth-$i$ vertices and their children), and because of the aforementioned
disconnected subtree property, there are no edges in $G \setminus T_i$
connecting two vertices in $T_i$.

Let $T'$ be whichever of $\bigcup_{i = 1}^W T_{2i}$ and $\bigcup_{i=1}^W T_{2i
+ 1}$ has the most edges, so $\abs{T'} = (W - 1)/2$. Now, $T'$ is a union of
stars, and by property~\ref{qmc:star}, for each star with degree $d$ there is
an assignment earning $\frac{d+1}{2}$ on it. Now, as these stars are disjoint,
and the edges in $G\setminus T'$ are all \emph{between} stars,
property~\ref{qmc:quarter} implies that there is an assignment to $G$ earning
\begin{align*}
\frac{\abs{T'} + 1}{2} + \frac{m - \abs{T'}}{4} &= \frac{m}{4} +
\frac{\abs{T'}}{8} + \frac{1}{2}\\
&\ge \frac{m}{4} + \frac{W - 1}{8} + \frac{1}{2}\\
&> \frac{m}{4} + \frac{W}{8}
\end{align*}
\end{proof}
energy, completing the proof.

\subsection{Estimating $W$ in the Stream}
In this section we give a simple classical algorithm for estimating $W$ in
logarithmic space.
\begin{restatable}{lemma}{wstream}
\label{lm:wstream}
Let $G$ be a weighted graph on $n$ vertices with weights that are multiples of
$1/\poly(n)$. Then for any $(\varepsilon, \delta) \in (0,1)$ there is a
streaming algorithm that returns an $\eps m$-additive approximation
to $W$ using $\bO{\frac{1}{\varepsilon^2}\log\frac{1}{\delta}\log n}$ space.
\end{restatable}

We will use Algorithm~\ref{alg:westimator} to obtain an estimator for $W$ which
we will then amplify.
\begin{algorithm}
\caption{A log-space algorithm for estimating $W$ in the stream.}
\label{alg:westimator}
\begin{algorithmic}
\State Sample $e$ from $E$ with probability proportional to $w_e$.
\State Sample $v$ from the two endpoints of $e$, each with probability $1/2$.
\If{$e$ is the last edge to arrive incident to $v$ in the stream}
\State $\Xb \gets 1$
\ElsIf{$w_e < w_f$ for some $f$ incident to $v$ that arrives \emph{after} $e$ in
the stream}
\State $\Xb \gets 0$
\Else
\State $w' \gets$ the greatest weight of any edge to arrive incident to $v$
after $e$ in the stream.
\State $\Xb \gets 1 - w'/w_e$
\EndIf
\State \textbf{return} $\Xb$
\end{algorithmic}
\end{algorithm}
\begin{lemma}
Algorithm~\ref{alg:westimator} can be implemented in the stream using $O(\log
n)$ space.
\end{lemma}
\begin{proof}
By using reservoir sampling, we can sample $e$ while only storing one edge (and
its weight) at a time. Then, whenever we have a candidate sample for $e$, we
can track the highest-weight edge to arrive incident to $v$ that arrives after
$e$ (throwing it out if our reservoir sampling procedure gives us a new
candidate for $e$). This requires tracking just two edges, each requiring
$O(\log n)$ space to store with its weight.  
\end{proof}

\begin{lemma}
\[
\E{\Xb} = W/2m
\]
\end{lemma}
\begin{proof}
Let $\Sc_{e,u}$ denote the event that the edge $e$ and its endpoint $u$ are
sampled, so $\prob{\Sc_{e,u}} = w_e/2m$. We can then write \[
\E{\Xb} = \sum_{u \in V}\sum_{v \in N(u)} \E{\Xb\middle| \Sc_{uv,u}} w_{uv}/2m 
\]
Now, for any vertex $v \in N(u)$ such that $w_{uv} < w_f$ for some $f$ that
appears after $uv$ in the stream and is incident to $u$, $\E{\Xb\middle|
\Sc_{uv,u}} = 0$. 

So let $N'$ be the set of vertices $v \in N(u)$ such that $w_{uv} \ge w_f$ for
all $f$ incident to $u$ that appear after $uv$ in the stream. Order $N'$ as
$(v_i)_{i=1}^k$ in the order they arrive in the stream, and let $w_i$ denote
the weight of $uv_i$. Then, $\E{\Xb\middle|
\Sc_{uv_k,u}} = 1$, while for each $i \in \brac{k-1}$,  $\E{\Xb\middle|
\Sc_{uv_i,u}} = 1 - w_{i+1}/w_{i}$. Therefore, \[
\sum_{i = 1}^k \E{\Xb\middle| \Sc_{uv_i,u}} w_{i}/2m = w_k/2m + \sum_{i =
1}^{k-1} (w_i - w_{i+1})/2m  = w_{1}/2m = \max_{v \in N(u)} w_{uv}/2m
\]
and so by summing over all $u \in V$ the result follows.
\end{proof}

\begin{lemma}
\[
\var(\Xb) \le 1
\]
\end{lemma}
\begin{proof}
This follows immediately from the fact that $\Xb$ takes values in the range
$\brac{0,1}$.
\end{proof}
We can now prove Lemma~\ref{lm:wstream}.
\wstream*
\begin{proof}
We have that $2m\Xb$ is an unbiased estimator for $W$, with variance $4m^2$. By
repeating Algorithm~\ref{alg:westimator} $\bT{\varepsilon^2}$ times in
parallel and averaging these estimators, we may obtain a variance
$\varepsilon^2 m^2/9$ estimator, and so by Chebyshev, with probability $2/3$ it
will be within $\varepsilon m$ of $W$. We may then repeat this entire process
$\bT{\log 1/\delta}$ times and take the median to obtain a $\varepsilon
m$-additive estimator of $W$ with probability $1 - \delta$.

Each iteration can be performed in parallel and requires $\bO{\log n}$ space,
completing the proof.
\end{proof}

\subsection{Completing the Upper Bound}
Theorem~\ref{thm:qmcub} now follows immediately.
\qmcub*
\begin{proof}
By the lemmas of~\ref{sec:qmcapprox}, it is sufficient to estimate $2m + W$ to
$(1 \pm \eps')$-multiplicative accuracy for a sufficiently small $\eps' =
\bT{\eps}$. $m$ can be calculated exactly with a single counter in $bO{\log
n}$ space. By Lemma~\ref{lm:wstream}, we can calculate $W$ to $\eps'
m$-additive accuracy in $\bO{\frac{1}{\eps^2}\log \frac{1}{\delta}}$ space.
\end{proof}

\section*{Acknowledgements}
The authors were supported by the Laboratory Directed Research and Development
program at Sandia National Laboratories, a multimission laboratory managed and
operated by National Technology and Engineering Solutions of Sandia, LLC., a
wholly owned subsidiary of Honeywell International, Inc., for the U.S.
Department of Energy’s National Nuclear Security Administration under contract
DE-NA-0003525. Also supported by the U.S. Department of Energy, Office of
Science, Office of Advanced Scientific Computing Research, Accelerated Research
in Quantum Computing program.
\newpage
\bibliographystyle{alpha}
\bibliography{refs}

\newcommand{\etalchar}[1]{$^{#1}$}
\begin{thebibliography}{AGMKS21}

\bibitem[AAV13]{aharonov2013guest}
Dorit Aharonov, Itai Arad, and Thomas Vidick.
\newblock Guest column: the quantum pcp conjecture.
\newblock {\em Acm sigact news}, 44(2):47--79, 2013.

\bibitem[AD21]{AD21}
Srinivasan Arunachalam and Jo{\~a}o~F Doriguello.
\newblock Matrix hypercontractivity, streaming algorithms and ldcs: the large
  alphabet case.
\newblock {\em arXiv preprint arXiv:2109.02600}, 2021.

\bibitem[AG09]{anh-guha}
K.~Ahn and S.~Guha.
\newblock Graph sparsification in the semi-streaming model.
\newblock {\em ICALP}, pages 328--338, 2009.

\bibitem[AGM20]{Anshu2020}
Anurag Anshu, David Gosset, and Karen Morenz.
\newblock {Beyond Product State Approximations for a Quantum Analogue of Max
  Cut}.
\newblock In Steven~T. Flammia, editor, {\em 15th Conference on the Theory of
  Quantum Computation, Communication and Cryptography (TQC 2020)}, volume 158
  of {\em Leibniz International Proceedings in Informatics (LIPIcs)}, pages
  7:1--7:15, Dagstuhl, Germany, 2020. Schloss Dagstuhl--Leibniz-Zentrum f{\"u}r
  Informatik.

\bibitem[AGMKS21]{A21}
Anurag Anshu, David Gosset, Karen~J Morenz~Korol, and Mehdi Soleimanifar.
\newblock Improved approximation algorithms for bounded-degree local
  hamiltonians.
\newblock {\em arXiv preprint arXiv:2105.01193}, 2021.

\bibitem[AZ19]{AZ19}
Dorit Aharonov and Leo Zhou.
\newblock Hamiltonian sparsification and gap-simulation.
\newblock In {\em 10th Innovations in Theoretical Computer Science Conference,
  {ITCS} 2019, January 10-12, 2019, San Diego, California, {USA}}, volume 124
  of {\em LIPIcs}, pages 2:1--2:21, 2019.

\bibitem[BARdW08]{BARdW08}
Avraham Ben-Aroya, Oded Regev, and Ronald de~Wolf.
\newblock A hypercontractive inequality for matrix-valued functions with
  applications to quantum computing and ldcs.
\newblock In {\em Proceedings of the 49th Annual IEEE Symposium on Foundations
  of Computer Science}, 2008.

\bibitem[BGKT19]{B19}
Sergey Bravyi, David Gosset, Robert K{\"o}nig, and Kristan Temme.
\newblock Approximation algorithms for quantum many-body problems.
\newblock {\em Journal of Mathematical Physics}, 60(3):032203, 2019.

\bibitem[BH16]{B16}
Fernando~GSL Brand{\~a}o and Aram~W Harrow.
\newblock Product-state approximations to quantum states.
\newblock {\em Communications in Mathematical Physics}, 342(1):47--80, 2016.

\bibitem[BMO{\etalchar{+}}15]{Barak2015Beating}
Boaz Barak, Ankur Moitra, Ryan O'Donnell, Prasad Raghavendra, Oded Regev, David
  Steurer, Luca Trevisan, Aravindan Vijayaraghavan, David Witmer, and John
  Wright.
\newblock {Beating the Random Assignment on Constraint Satisfaction Problems of
  Bounded Degree}.
\newblock In Naveen Garg, Klaus Jansen, Anup Rao, and Jos{\'e} D.~P. Rolim,
  editors, {\em Approximation, Randomization, and Combinatorial Optimization.
  Algorithms and Techniques (APPROX/RANDOM 2015)}, volume~40 of {\em Leibniz
  International Proceedings in Informatics (LIPIcs)}, pages 110--123, Dagstuhl,
  Germany, 2015. Schloss Dagstuhl--Leibniz-Zentrum fuer Informatik.

\bibitem[CGS{\etalchar{+}}22]{CGSVV22}
Chi-Ning Chou, Alexander Golovnev, Madhu Sudan, Ameya Velingker, and
  Santhoshini Velusamy.
\newblock Linear space streaming lower bounds for approximating {CSPs}.
\newblock In {\em STOC{\textquoteright}22}, 2022.

\bibitem[CM16]{Cubitt16Complexity}
Toby Cubitt and Ashley Montanaro.
\newblock Complexity classification of local hamiltonian problems.
\newblock {\em SIAM Journal on Computing}, 45(2):268--316, 2016.

\bibitem[CW04]{CW04}
M.~Charikar and A.~Wirth.
\newblock Maximizing quadratic programs: extending grothendieck's inequality.
\newblock In {\em 45th Annual IEEE Symposium on Foundations of Computer
  Science}, pages 54--60, 2004.

\bibitem[DM20]{DM20}
Jo{\~a}o~F. Doriguello and Ashley Montanaro.
\newblock {Exponential Quantum Communication Reductions from Generalizations of
  the Boolean Hidden Matching Problem}.
\newblock In {\em 15th Conference on the Theory of Quantum Computation,
  Communication and Cryptography (TQC 2020)}, volume 158, pages 1:1--1:16,
  2020.

\bibitem[FGG14]{farhi2014quantum}
Edward Farhi, Jeffrey Goldstone, and Sam Gutmann.
\newblock A quantum approximate optimization algorithm.
\newblock {\em arXiv preprint arXiv:1411.4028}, 2014.

\bibitem[GK12]{G12}
Sevag Gharibian and Julia Kempe.
\newblock Approximation algorithms for qma-complete problems.
\newblock {\em SIAM Journal on Computing}, 41(4):1028--1050, 2012.

\bibitem[GKK{\etalchar{+}}07]{GKKRd07}
Dmitry Gavinsky, Julia Kempe, Iordanis Kerenidis, Ran Raz, and Ronald de~Wolf.
\newblock Exponential separations for one-way quantum communication complexity,
  with applications to cryptography.
\newblock In {\em Proceedings of the Thirty-ninth Annual ACM Symposium on
  Theory of Computing}, STOC '07, pages 516--525, New York, NY, USA, 2007. ACM.

\bibitem[GKK{\etalchar{+}}08]{GKKRW07}
Dmitry Gavinsky, Julia Kempe, Iordanis Kerenidis, Ran Raz, and Ronald de~Wolf.
\newblock Exponential separation for one-way quantum communication complexity,
  with applications to cryptography.
\newblock {\em {SIAM} J. Comput.}, 38(5):1695--1708, 2008.

\bibitem[GP19]{GharibianParekh19}
Sevag Gharibian and Ojas Parekh.
\newblock {Almost Optimal Classical Approximation Algorithms for a Quantum
  Generalization of Max-Cut}.
\newblock In Dimitris Achlioptas and L{\'a}szl{\'o}~A. V{\'e}gh, editors, {\em
  Approximation, Randomization, and Combinatorial Optimization. Algorithms and
  Techniques (APPROX/RANDOM 2019)}, volume 145 of {\em Leibniz International
  Proceedings in Informatics (LIPIcs)}, pages 31:1--31:17, Dagstuhl, Germany,
  2019. Schloss Dagstuhl--Leibniz-Zentrum fuer Informatik.

\bibitem[GW95]{GoemansWilliamson95}
Michel~X. Goemans and David~P. Williamson.
\newblock Improved approximation algorithms for maximum cut and satisfiability
  problems using semidefinite programming.
\newblock {\em J. ACM}, 42(6):1115--1145, nov 1995.

\bibitem[GY21]{GY21}
Gregory Gutin and Anders Yeo.
\newblock Lower bounds for maximum weighted cut.
\newblock {\em arXiv preprint arXiv:2104.05536}, 2021.

\bibitem[HLP20]{H20}
Sean Hallgren, Eunou Lee, and Ojas Parekh.
\newblock An approximation algorithm for the {MAX-2-Local Hamiltonian} problem.
\newblock In {\em Approximation, Randomization, and Combinatorial Optimization.
  Algorithms and Techniques (APPROX/RANDOM 2020)}. Schloss
  Dagstuhl-Leibniz-Zentrum f{\"u}r Informatik, 2020.

\bibitem[HM19]{HM19}
Yassine Hamoudi and Fr{\'e}d{\'e}ric Magniez.
\newblock {Quantum Chebyshev's Inequality and Applications}.
\newblock In {\em 46th International Colloquium on Automata, Languages, and
  Programming (ICALP 2019)}, volume 132 of {\em Leibniz International
  Proceedings in Informatics (LIPIcs)}, pages 69:1--69:16, Dagstuhl, Germany,
  2019. Schloss Dagstuhl--Leibniz-Zentrum fuer Informatik.

\bibitem[HNP{\etalchar{+}}21]{hwang2021unique}
Yeongwoo Hwang, Joe Neeman, Ojas Parekh, Kevin Thompson, and John Wright.
\newblock Unique games hardness of quantum max-cut, and a vector-valued
  borell's inequality.
\newblock {\em arXiv preprint arXiv:2111.01254}, 2021.

\bibitem[JN14]{JN14}
Rahul Jain and Ashwin Nayak.
\newblock The {Space} {Complexity} of {Recognizing} {Well}-{Parenthesized}
  {Expressions} in the {Streaming} {Model}: {The} {Index} {Function}
  {Revisited}.
\newblock {\em IEEE Transactions on Information Theory}, 60(10):6646--6668,
  October 2014.

\bibitem[Kal22]{K21}
John Kallaugher.
\newblock A quantum advantage for a natural streaming problem.
\newblock In {\em 2021 IEEE 62nd Annual Symposium on Foundations of Computer
  Science (FOCS)}, pages 897--908, 2022.

\bibitem[KK15]{KK15}
Dmitry Kogan and Robert Krauthgamer.
\newblock Sketching cuts in graphs and hypergraphs.
\newblock In {\em Proceedings of the 2015 Conference on Innovations in
  Theoretical Computer Science}, pages 367--376, 2015.

\bibitem[KK19]{KK19}
Michael Kapralov and Dmitry Krachun.
\newblock An optimal space lower bound for approximating max-cut.
\newblock In {\em Proceedings of the 51st Annual ACM SIGACT Symposium on Theory
  of Computing}, STOC 2019, pages 277--288, New York, NY, USA, 2019.
  Association for Computing Machinery.

\bibitem[KKL88]{KKL88}
J.~Kahn, G.~Kalai, and N.~Linial.
\newblock The influence of variables on boolean functions.
\newblock In {\em Proceedings of the 29th Annual Symposium on Foundations of
  Computer Science}, SFCS '88, pages 68--80, Washington, DC, USA, 1988. IEEE
  Computer Society.

\bibitem[KKMO07]{KhotKMO07}
Subhash Khot, Guy Kindler, Elchanan Mossel, and Ryan O'Donnell.
\newblock Optimal inapproximability results for max-cut and other 2-variable
  csps?
\newblock {\em SIAM J. Comput.}, 37(1):319--357, 2007.

\bibitem[KKP18]{KKP18}
John Kallaugher, Michael Kapralov, and Eric Price.
\newblock The sketching complexity of graph and hypergraph counting.
\newblock In {\em 2018 IEEE 59th Annual Symposium on Foundations of Computer
  Science (FOCS)}, pages 556--567. IEEE, 2018.

\bibitem[KKS15]{KKS15}
Michael Kapralov, Sanjeev Khanna, and Madhu Sudan.
\newblock Streaming lower bounds for approximating {MAX-CUT}.
\newblock In {\em 26th ACM-SIAM Symposium on Discrete Algorithms (SODA)}, 2015.

\bibitem[KKSV17]{KKSV17}
Michael Kapralov, Sanjeev Khanna, Madhu Sudan, and Ameya Velingker.
\newblock $(1 + {\Omega}(1))$-{A}pproximation to {MAX-CUT} requires linear
  space.
\newblock In {\em Proceedings of the Twenty-Eighth Annual {ACM-SIAM} Symposium
  on Discrete Algorithms, {SODA} 2017, Barcelona, Spain, Hotel Porta Fira,
  January 16-19}, pages 1703--1722, 2017.

\bibitem[Kla07]{K07}
Hartmut Klauck.
\newblock Lower bounds for quantum communication complexity.
\newblock {\em SIAM Journal on Computing}, 37(1):20--46, 2007.

\bibitem[LG06]{LG06}
Fran\c{c}ois Le~Gall.
\newblock Exponential separation of quantum and classical online space
  complexity.
\newblock In {\em Proceedings of the eighteenth annual {ACM} symposium on
  {Parallelism} in algorithms and architectures}, {SPAA} '06, pages 67--73, New
  York, NY, USA, July 2006. Association for Computing Machinery.

\bibitem[Mon16]{M16}
Ashley Montanaro.
\newblock The quantum complexity of approximating the frequency moments.
\newblock {\em Quantum Info. Comput.}, 16(13--14):1169--1190, October 2016.

\bibitem[NT17]{NT17}
Ashwin Nayak and Dave Touchette.
\newblock Augmented index and quantum streaming algorithms for dyck(2).
\newblock In {\em Proceedings of the 32nd Computational Complexity Conference},
  CCC '17, Dagstuhl, DEU, 2017. Schloss Dagstuhl--Leibniz-Zentrum fuer
  Informatik.

\bibitem[NV18]{natarajan2018low}
Anand Natarajan and Thomas Vidick.
\newblock Low-degree testing for quantum states, and a quantum entangled games
  pcp for qma.
\newblock In {\em 2018 IEEE 59th Annual Symposium on Foundations of Computer
  Science (FOCS)}, pages 731--742. IEEE, 2018.

\bibitem[O'D14]{O14}
Ryan O'Donnell.
\newblock {\em Analysis of {Boolean} functions}.
\newblock Cambridge University Press, 2014.

\bibitem[PM15]{P15}
Stephen Piddock and Ashley Montanaro.
\newblock The complexity of antiferromagnetic interactions and 2d lattices.
\newblock {\em arXiv preprint arXiv:1506.04014}, 2015.

\bibitem[PT21a]{PT21}
Ojas Parekh and Kevin Thompson.
\newblock {Application of the Level-2 Quantum Lasserre Hierarchy in Quantum
  Approximation Algorithms}.
\newblock In {\em 48th International Colloquium on Automata, Languages, and
  Programming (ICALP 2021)}, volume 198 of {\em Leibniz International
  Proceedings in Informatics (LIPIcs)}, pages 102:1--102:20, 2021.

\bibitem[PT21b]{PT20}
Ojas Parekh and Kevin Thompson.
\newblock {Beating Random Assignment for Approximating Quantum 2-Local
  Hamiltonian Problems}.
\newblock In {\em 29th Annual European Symposium on Algorithms (ESA 2021)},
  volume 204 of {\em Leibniz International Proceedings in Informatics
  (LIPIcs)}, pages 74:1--74:18, 2021.

\bibitem[Raz95]{R95}
Ran Raz.
\newblock Fourier analysis for probabilistic communication complexity.
\newblock {\em Computational Complexity}, 5(3):205--221, 1995.

\bibitem[SW12]{SW12}
Yaoyun Shi and Xiaodi Wu.
\newblock Limits of quantum one-way communication by matrix hypercontractive
  inequality.
\newblock 2012.

\bibitem[VY11]{VY11}
Elad Verbin and Wei Yu.
\newblock The streaming complexity of cycle counting, sorting by reversals, and
  other problems.
\newblock {\em SODA}, pages 11--25, 2011.

\end{thebibliography}
\end{document}